\documentclass[DIV12]{scrartcl}
\usepackage{cmap}
\usepackage[T1]{fontenc}
\usepackage[utf8]{inputenc}
\usepackage{lmodern}
\usepackage{microtype}
\usepackage{cellspace}
\usepackage{xspace}
\usepackage{amsmath,amssymb,amsthm}
\usepackage{url}
\usepackage{paralist}
\usepackage{booktabs}
\newtheorem{theorem}{Theorem}[section]

\newtheorem{proposition}[theorem]{Proposition}

\newtheorem{lemma}[theorem]{Lemma}

\newtheorem{claim}[theorem]{Claim}

\newtheorem{defn}[theorem]{Definition}
\newtheorem{expl}[theorem]{Example}
\newenvironment{definition}{\begin{defn}\upshape}{\hfill\qed\end{defn}}
\newenvironment{example}{\begin{expl}\upshape}{\hfill\qed\end{expl}}

\newcommand{\pspace}{\complexity{PSPACE}}
\newcommand{\expspace}{\complexity{PSPACE}}
\newcommand{\twoexp}{\complexity{2EXPTIME}}
\newcommand{\threeexp}{\complexity{3EXPTIME}}
\newcommand{\twonexp}{\complexity{2NEXPTIME}}
\newcommand{\cotwonexp}{\complexity{co2NEXPTIME}}

\newcommand{\empint}{\kw{EmpOffAcc}}
\newcommand{\manint}{\kw{EmpManAcc}}

\newcommand{\offint}{\kw{OfficeInfoAcc}}
\newcommand{\appint}{\kw{StateApprAcc}}

\newcommand{\kw}[1]{{\mathsf{#1}}\xspace}
\newcommand{\sch}{\kw{Sch}}
\newcommand{\tables}{\kw{Tables}}
\newcommand{\att}{\kw{ATT}}
\newcommand{\dom}{\kw{Dom}}
\newcommand{\adom}{\kw{Adom}}

\newcommand{\acs}{\kw{ACS}}

\newcommand{\anacm}{\kw{AcM}}
\newcommand{\sof}{\kw{Rel}}
\newcommand{\inat}{\kw{InputAtt}}
\newcommand{\adb}{\kw{I}}
\newcommand{\aninst}{\kw{I}}

\newcommand{\abind}{\kw{Bind}}

\newcommand{\acfg}{\kw{Conf}}

\newcommand{\aresp}{\kw{Resp}}

\newcommand{\complexity}[1]{\textsf{\mdseries\upshape #1}}

\newcommand{\set}[1]{\{\,#1\,\}}

\renewcommand{\phi}{\varphi}
\renewcommand{\epsilon}{\varepsilon}
\renewcommand{\leq}{\leqslant}

\newcommand{\np}{\complexity{NP}}
\newcommand{\conp}{\complexity{coNP}}

\newcommand{\sigmatp}{\complexity{$\Sigma_2^{\text{P}}$}}
\newcommand{\pitp}{\complexity{$\Pi_2^{\text{P}}$}}
\newcommand{\nexptime}{\complexity{NEXPTIME}}
\newcommand{\conexptime}{\complexity{coNEXPTIME}}
\newcommand{\acz}{\complexity{$\text{AC}^0$}}

\newcommand{\lti}{\kw{LTR}}

\newcommand{\ii}{\kw{IR}}

\renewcommand{\succ}{\kw{Succ}}

\newcommand{\nextcell}{\relation{NextCell}}
\newcommand{\correct}{\relation{Correct}}
\newcommand{\nexveraddrbit}{\relation{NextVAddrBit}}
\newcommand{\nexhoraddrbit}{\relation{NextHAddrBit}}
\newcommand{\sameaddressas}{\relation{SameAddressAs}}
\newcommand{\notsucc}{\relation{NotSucc}}
\newcommand{\color}{\relation{Color}}
\newcommand{\belowbit}{\relation{BelowBit}}
\newcommand{\iszero}{\relation{IsZero}}
\newcommand{\isfirstzero}{\relation{IsFirstZero}}
\newcommand{\isone}{\relation{IsOne}}
\newcommand{\notsamebit}{\relation{NotSameBit}}
\newcommand{\isbit}{\relation{IsBit}}
\newcommand{\sucbit}{\relation{SuccessorBit}}

\renewcommand{\vec}[1]{\mathbf{#1}}

\newcommand{\bb}{{\vec b}}
\newcommand{\cc}{{\vec c}}
\newcommand{\dd}{{\vec d}}
\newcommand{\ee}{{\vec e}}
\newcommand{\fff}{{\vec f}}
\renewcommand{\ggg}{{\vec g}}
\newcommand{\vv}{{\vec v}}
\newcommand{\hh}{{\vec h}}

\let\relation=\mathit

\newcommand*{\tile}{\relation{Tile}}
\newcommand*{\tiletype}{\relation{TileType}}
\newcommand*{\sametile}{\relation{SameTile}}
\newcommand*{\rboolean}{\relation{Bool}}
\newcommand*{\horiz}{\relation{Horiz}}
\newcommand*{\vertic}{\relation{Vert}}
\newcommand*{\rand}{\relation{And}}
\newcommand*{\ror}{\relation{Or}}
\newcommand*{\eq}{\relation{Eq}}
\newcommand*{\req}{\eq}
\newcommand*{\SUCC}{\relation{SUCC}}
\newcommand*{\SUB}{\relation{SUB}}
\newcommand*{\BOOLCONS}{\relation{BOOLCONS}}

\newcommand{\chain}{\kw{Chain}}

\usepackage{hyperref}

\title{Determining Relevance of Accesses at Runtime\\(Extended Version)\thanks{This paper is
an extended version of the conference
article~\cite{conference_version}.}}
\author{
Michael Benedikt\\
       {Computing Laboratory}\\
                {Oxford University}\\
       {\makebox{Oxford OX1 3QD, UK}}\\
       {\normalsize\sf\makebox{michael.benedikt@comlab.ox.ac.uk}}
\and Georg Gottlob\\
       {Computing Laboratory}\\
       {\& Oxford Man Institute}\\
                {Oxford University}\\
       {\normalsize\sf\makebox{georg.gottlob@comlab.ox.ac.uk}}
\and Pierre Senellart\\
       {\makebox{Institut Télécom; Télécom
ParisTech}}\\
       {CNRS LTCI, 46 rue Barrault}\\
       {75634 Paris, France}\\
       {\normalsize\sf\makebox{pierre.senellart@telecom-paristech.fr}}
}
\date{}

\begin{document}

\maketitle

\begin{abstract}
Consider the situation where a query is to be answered
using Web sources that restrict the accesses that can be made 
on backend relational data by requiring some attributes to be given as
input of the service.
The accesses provide lookups on the collection
of attributes values that match the binding. 
They can differ in whether or not they require arguments to 
be generated from prior accesses. 
Prior work has focused on the question of whether a query can be answered
using a set of data sources, and in developing static access plans (e.g.,
Datalog programs) that implement query answering.
We are interested in dynamic aspects of the query answering problem: given 
partial information about the data, which accesses could provide  relevant
data for answering a given query?
We consider immediate and long-term
notions of  ``relevant accesses'', and ascertain the complexity
of query relevance, for both conjunctive queries and arbitrary
positive queries.
In the process, we relate dynamic relevance of an
access to query containment under access limitations and characterize the
complexity of this problem; we produce several complexity results about
containment that are of interest by themselves.
\end{abstract}

\section{Introduction} \label{sec:intro}

\paragraph*{Relevance under access limitations}

A large part of the information on the World Wide Web is not available
through the \emph{surface Web}, the set of Web pages reachable by
following hyperlinks, but lies in the \emph{deep Web} (or \emph{hidden
Web}), that provides
entry points to databases accessible via HTML forms or Web services.
Hundreds of thousands of such deep-Web sources exist~\cite{he2007accessing}. Even when the information is available on
the surface Web, it can be more effective to access it through a
(possibly elaborate) Web form query. Each source of the deep Web has one
or several interfaces that limit the kind of accesses that can be
performed, e.g., some fields of the forms have to be filled in before
submission.

A number of works, e.g.~\cite{rajaraman1995answering,duschka1997recursive},
have dealt with the problems of answering queries using
views in the presence of such access restrictions
but the focus is usually on obtaining a static query plan (e.g., 
a rewriting of an original conjunctive query, or a Datalog program).
We consider a dynamic approach to query answering and study the
following problem: given some existing
knowledge about the data, knowledge that is bound to evolve as we access
sources, is making this particular access relevant to a query? In other
words, can this particular access give, immediately or after some other
accesses, some knowledge that will yield an answer to the query?

Let us consider the following example.
A user wishes to get  information about
the loan capabilities of a large bank.
A relation schema for this can be visualized as:
\begin{verbatim}
Employee(EmpId, Title, LastName, FirstName, OffId)
Office(OffId, StreetAddress, State, Phone)
Approval(State, Offering)
Manager(EmpId, EmpId)
\end{verbatim}

\texttt{Employee} stores information about employees, including their
title and office. \texttt{Office} stores information about offices,
including the state in which they are located. \texttt{Approval} tells which
kinds of loans a bank is approved to make in each state, while
\texttt{Manager} stores which employee manages which other employee.

Data from a number of distinct Web data sources (or distinct query
interfaces from the same source) can be used to answer the query:
\begin{compactitem}
\item a form $\empint$ where an \texttt{EmpId}  can be entered,
which returns office records for that employee;
\item a form $\manint$ where an \texttt{EmpId} can be entered, and the identifiers of their managers
are returned;
\item a form $\offint$ that allows one to enter an \texttt{OffId},
and returns all
the office information;
\item a form $\appint$ that allows one to enter a state, and returns the approval information
for that state.
\end{compactitem}

A user wishes to know
if there is
some  loan officer
for  their bank
located in Illinois,
and also whether the company is authorized to perform 30-year mortgages in Illinois.
This can be phrased as a Boolean query $Q$, expressible in SQL as follows:

\begin{verbatim}
SELECT DISTINCT 1
FROM Employee E, Office O, Approval A
WHERE E.Title='loan officer' AND E.OffId=O.OffId
  AND O.State='Illinois'     AND A.State='Illinois'
  AND A.Offering='30yr'
\end{verbatim}

A federated query engine that tries to
answer this query will require
querying the
distinct interfaces with concrete values.  At a certain stage
the query engine may know about a particular set of employee names, and
also know certain facts -- either via querying or via an additional
knowledge base. Which
interfaces should it use to answer the query?
In particular: Is an access to the $\manint$ form with \texttt{EmpId}
``12345'' useful for answering $Q$? There are actually a number of
subtleties in this question, that we discuss now.

\emph{The relevance of an access depends on the existing knowledge
base.}
At the beginning of the process, when no other information is known
about the data, the access might be useful to get some other
\texttt{EmpId} which may in turn be used in the $\empint$ interface to
find an Illinoisan loan officer.
But if we already know  that the company has a loan officer
located in Illinois, then clearly such an access is unnecessary.
We call this existing knowledge base the \emph{configuration} in which
the access is made.

\emph{The relevance depends on  how closely linked the Web forms are.}
Clearly the  interface is irrelevant to the query if the query engine
is free to enter \texttt{EmpId} values ``at random'' into the $\empint$ interface.
But if such values are widely dispersed and there is no way to guess them,
an efficient tactic might be to take \texttt{EmpId}'s that we know about, query
for their managers using the $\manint$ interface, and then use the resulting
offices in the $\offint$ interface.
In this work we will thus distinguish between accesses that require
a value that is already in the knowledge base of the engine
(\emph{dependent} accesses), from those
that allow a ``free guess''. Note that in the case of {\it static} query
answering plans, the notion of a ``free access'' trivializes the questions.

\emph{The relevance depends on whether one is interested in immediate or
long-term impact.}
Without any initial knowledge, there is no way an access to $\empint$ may
directly provide a witness for the query. On the other hand, as discussed
above, the result to an access may be used in some cases
to gain some information that will ultimately satisfy the query. In this
work, we consider both \emph{immediate relevance} and \emph{long-term
relevance} of a particular access.

\bigskip

\paragraph*{Main questions studied}

In this article, we are interested in the following problems:
\begin{compactenum}[(i)]
\item How to define a model for querying under access restrictions that
takes into account the history of accesses?

\item What is the complexity of relevance?

\item Calì and Martinenghi have studied in~\cite{cali2008conjunctive}
the complexity of \emph{containment under access constraints},
motivated by query optimization. How does relevance relate to containment?
Are these notions at all related, and if so, 
can the respective decision problems be transformed into one
another?

\item What is the complexity of containment under access constraints?

\item If problems are hard, can we identify the source of this complexity?
\end{compactenum}

One particular reason why these problems are challenging is that they do
not deal with a concrete database, but a virtual database of which we have a
partial view, a view that evolves as we access it. The notion of
relevance of accesses has not been investigated in the literature; the
closest work, on containment under access
constraints~\cite{cali2008conjunctive}, only provides an upper bound of
\conexptime, for
a restricted query language (conjunctive queries, with only limited
use of constants). Determining a lower
bound for containment was left as an open problem. Hardness results are
difficult to obtain, because the access model that we present is quite
simple and does not offer obvious clues of how to encode known hard
problems to get lower bounds.

\paragraph*{Results}
We emphasize the following contributions of
our work.

\begin{table*}
\caption{Summary of combined complexity results}
\label{tab:complexity-results}
\hspace*{-2em}
\begin{tabular}{lScScSc}
\toprule
&Immediate relevance & Long-term relevance & Containment\\
&&(Boolean access)\\
\midrule
Independent accesses (CQs)&
\complexity{DP}-complete&\sigmatp-complete& \pitp-complete\\
Independent accesses (PQs)&
\complexity{DP}-complete&\sigmatp-complete&\pitp-complete\\
Dependent accesses (CQs)&
\complexity{DP}-complete&\nexptime-complete&\conexptime-complete\\
Dependent accesses (PQs)&
\complexity{DP}-complete&\twonexp-complete&\cotwonexp-complete\\
\bottomrule
\end{tabular}
\end{table*}

We provide the first formal definition 
of dynamic relevance of accesses for a query $Q$, using a simple and
powerful model, answering thus item~(i).

We give a combined complexity characterization of the
relevance problem in all combinations of cases
(immediate or long-term relevance, independent or dependent
accesses, conjunctive or positive queries), inside the polynomial
and exponential hierarchy of complexity classes; for long-term relevance,
we mostly focus on accesses without any input, extension to arbitrary
accesses is left for future work. This gives a
satisfactory answer to question~(ii).
For several of our hardness  results,
we invented sophisticated coding techniques to enforce database accesses
to produce grids that would then allow us to encode tiling problems. One
particular  hurdle to overcome was the limited ``coding power'' of
conjunctive queries. We therefore had to use and extend techniques for
encoding disjunctions into a conjunctive query.

We exhibit reductions in both directions between dynamic relevance and
containment under access constraints.
By these results, we succeed in elucidating the
relationship between containment and long term relevance, thus providing
an exhaustive answer to item~(iii).

We generalize the \conexptime{} upper bound to a
stronger notion of containment, and provide a matching lower bound,
solving thus item~(iv).
    This \conexptime{} upper bound for containment, and the associated
\nexptime{} upper bound for 
    relevance are rather surprising and not at all obvious. In fact, the
more immediate upper bounds, that we show for positive queries, are  
    \cotwonexp{} and  \twonexp, respectively. 

We highlight specific cases of interest
where the complexity of relevance
is lower, e.g., conjunctive queries with a single occurrence of a
relation, or conjunctive queries with small arity.
We also show that all problems are polynomial-time in data complexity
(for the independent case, \acz),
suggesting the feasibility of the relevance analysis. These two points
together bring a first answer to item~(v).

A summary of complexity results is shown in
Table~\ref{tab:complexity-results}.

\paragraph*{Organization}
We start with formal definitions of the problem and the
various degrees of relevance of a query to an access in
Section~\ref{sec:prelim}. We next establish
(Section~\ref{sec:containment}) the connection between
relevance and the topic of containment that was studied
in~\cite{cali2008conjunctive}. In Section~\ref{sec:indep}, we study the
case of independent accesses (accesses that do not require the input
value having been generated by a previous access). Here
the access patterns play quite a small role, but relevance
is still a non-trivial notion --  the issues
revolve around reasoning about a very restricted form of query containment.
In Section~\ref{sec:dep} we turn to dependent accesses, where the notion
of containment is of primary interest.
We extend techniques of \cite{cali2008conjunctive} to isolate the complexity of containment
under access patterns
for both conjunctive queries and positive queries; in the process
we give the complexity of relevance for both these classes.
We then present some particular tractable cases, when relations are
assumed to have small arity, in Section~\ref{sec:psp}.
Related work is
discussed in Section~\ref{sec:related}.

\section{Preliminaries} \label{sec:prelim}
 
We use bold face (e.g., $\vec a$)
to denote sets of
attributes or tuples of constants.

\paragraph*{Modeling data sources} We have a schema $\sch$ consisting of
a set of  relations $\tables(\sch)=\{S_1 \ldots S_n\}$, each $S_i$ having
a set of attributes $\att(S_i)$.
Following~\cite{li2001answering,cali2008conjunctive}, we assume each attribute $a_{ij}$ of relation
$S_i$ has an \emph{abstract domain} $\dom(a_{ij})$ chosen in some
countable
set of abstract domains.
Two attributes may share the same
domain and different domains may overlap. In the dependent case, domains
are used to constrain some input values to come from constants of the
appropriate type.

Given a source instance $\adb$ for $\sch$, a {\it configuration for $\adb$},  (with respect to $\sch$,
when not understood from context)
is a subset $\acfg$ of $\adb$, that is, for each 
$S_i$, a subset $\acfg(S_i)$ of the tuples in $\adb(S_i)$ (the content of
relation $S_i$ in $\adb$).
By a \emph{configuration} we mean any $\acfg$ that is a configuration for some instance $\adb$.
We then say that a configuration $\acfg$ is \emph{consistent} with $\adb$ if
$\acfg\subseteq\adb$.
Note that a configuration will generally be consistent with many
instances (in particular, the empty configuration is consistent with all
instances).

We have a set of \emph{access methods} $\acs=\set{\anacm_1 \ldots
\anacm_m}$ with each $\anacm_i$ consisting of a source relation
$\sof(\anacm_i)$ and a set $\inat(\anacm_i)$ of input attributes from the
set of attributes of $\sof(\anacm_i)$;
implicitly, each access method allows one to put
in a tuple of values for $\inat(\anacm_i)$ and get as a result a set of
matching tuples. If a relation does not have any access methods, no new
facts can be learned about this relation:
its content is fixed as that of the initial configuration.

Access methods are of two different varieties, based on the values that can
be entered into them
An access method may
be either \emph{dependent} or \emph{independent}. In a dependent access,
one can only use as an input bindings values
that have appeared in the configuration in the appropriate domain.
An independent access
can make use of any value. 

A combination of an access method and a binding to the input places
of the accessed relation will be referred to as an {\it access}.
We will often write an access by adding ``$?$'' to the non-input places, omitting
the exact method:
e.g. $R(3,?)$ is an access (via some method) to $R$ with the first place
bound to $3$. If $R$ does not have any output attributes, we say that it
is a \emph{Boolean} access, and we write for instance $R(3)?$ for an access
that checks whether $3 \in R$. If $R$ does not have any input attributes,
we say that it is a \emph{free} access.
We do not assume access methods to be \emph{exact}, i.e., to return all
tuples that are compatible with the binding. They are only assumed to be
\emph{sound}, i.e., they can return any sound subset of the data, and
possibly a different subset on each use.

Given a set of attributes $\vec a$ of a relation $S_i$,
 a  database instance  $\adb$,
and a binding $\abind$ of
each attribute in $\vec a$ to a value from $\dom(\vec a)$,
we let $\adb(\abind,S_i)$ to
be the set of tuples in $\adb$ whose projection onto $\vec a$ agrees with $\abind$.
For a configuration $\acfg$, its  \emph{active domain} 
$\adom(\acfg)=\set{(c,\mathcal C)}$ is the set of constants
that appear in a $\acfg(S_i)$ for some $i$, together with their abstract
domains:
for instance,
if $(c,d)\in\acfg(S)$ and $\dom(\att(S))=(\mathcal C,\mathcal D)$, both
$(c,\mathcal C)$ and
$(d,\mathcal D)$ belong to $\adom(\acfg)$.

Given a configuration $\acfg$, a \emph{well-formed access}
consists of an access method $\anacm$ 
and an assignment 
$\abind$ of values 
to the  attributes of $\inat(\anacm)$
such that  either a)~$\anacm$ is independent;
or b)~$\anacm$ is dependent and all values in
$\abind$, together with corresponding domains of the input attributes, are
in $\adom(\acfg)$.
A well-formed access $(\anacm, \abind)$ at configuration $\acfg$ on instance $\adb$
leads, possibly non-deterministically,  to any new configuration $\acfg'$ in which:

\begin{compactenum}[(i)]
\item 
\(
\acfg(\sof(\anacm)) \subseteq \acfg'(\sof(\anacm))\);
\item \(\acfg'(\sof(\anacm) 
\subseteq \acfg(\sof(\anacm))
\cup \adb(\abind,\sof(\anacm))
\);
\item $\acfg(S_i) = \acfg'(S_i)$ for all $S_i \neq \sof(\anacm)$.
\end{compactenum}
That is, the tuples seen in $S_i=\sof(\anacm)$ can increase by adding some tuples consistent
with the access, the access $(\anacm, \abind)$ is now completed 
and every other access stays the same in terms of completion. Note that the new
configuration is still consistent with the instance.

In general, there can be many successor configurations. We sometimes
write $\acfg+(\anacm,\abind, \aresp)$ 
to denote
an arbitrary such ``response configuration''.

A configuration $\acfg'$ is \emph{reachable} from another configuration
$\acfg$ (w.r.t.\ an instance)
if there is some sequence of well-formed accesses that can lead from
$\acfg$ to $\acfg'$.

\paragraph*{Queries}
We will consider conjunctive queries (CQs), i.e., conjunctions of atomic
facts, and positive existential queries, or just positive queries (PQs)
for short,
i.e., first-order formulas without universal quantifiers or negation.
PQs have the inconvenient of being
\emph{unsafe}~\cite{AbiteboulHV95} query languages; however, as discussed
at the end of this section, we focus on Boolean queries in this work,
for which the problem does not occur.
We recall some basic facts about the complexity of these languages:
query evaluation over CQs or PQs is
\np-complete in combined complexity
(membership in \np{} holds for any existentially quantified
first-order query, \np-hardness is a classical
result~\cite{chandra1977optimal}), while the data complexity of
evaluating an arbitrary first-order query is {\acz}
~\cite{AbiteboulHV95}.
On the
other hand, the
query containment problem is \np-complete for
CQs~\cite{chandra1977optimal}, but it is
\pitp-complete for PQs~\cite{SagivY80}. We require that variables
shared across subgoals of a query are consistent with domain
restrictions: if the same variables $x$ occur in attribute $a$ of $R$
and attribute $a'$ of $R'$ then $\dom(a)=\dom(a')$. The \emph{output
domain} of a query $Q$ is the tuple of domains of the output variables of
the query.
We also assume that all constants appearing in the query are present in the
configuration; in this way, constants from the query can be used in
dependent accesses.

The fundamental question we ask in this work is:
given a configuration $\acfg$, which well-formed accesses
for that configuration can contribute to answering the query $Q$?

\paragraph*{Immediate relevance}

We begin with analyzing whether a given access can have immediate impact on a
query -- whether the result of the access
can impact the information we have about a query output.

We recall the notion of  certain answers, which capture the notion of ``information'' precisely.
Given a configuration $\acfg$ and a tuple $\vec t$ of constants from
$\acfg$ with the same domain as the output domain of a query $Q$, we
say that 
$\vec t$ is a \emph{certain answer for $Q$ at $\acfg$} if for every
instance~$\adb$ consistent with $\acfg$ 
we have  $\vec t \in Q(\adb)$. If the query $Q$ is Boolean (i.e., with no
free variables), we say that it is certain (or simply true)
in a configuration 
$\acfg$ if for every instance~$\aninst$ consistent with $\acfg$, $Q(\aninst)$ is true.

We now consider the impact of a new well-formed access $(\anacm, \abind)$
on source $S$ in
a configuration $\acfg$. The result of this is some new set of tuples
$\aresp$ for $S$.

Let 
$\acfg+ (\anacm, \abind, \aresp)$ be a response configuration for
the access $(\anacm, \abind)$. We say the configuration (or even
the response
$\aresp$, seen as a collection of tuples) is
an \emph{increasing response} for $Q$ to $(\anacm,\abind)$ in $\acfg$ if
there exists a tuple $\vec{t}$ such that
$\vec{t}$ is not a certain answer for $Q$ at $\acfg$ while $\vec{t}$ is a
certain answer for $Q$ at $\acfg + (\anacm, \abind, \aresp)$.

An access $(\anacm,\abind)$ is \emph{immediately relevant} for the query
$Q$ ($\ii$
 in short) in a configuration
$\acfg$ if there is
some increasing response to the access.

\paragraph*{Long-term impact}
We formalize the
notion of  an access that can \emph{eventually} yield information.

Given an  access $(\anacm,\abind)$, a \emph{path} from $(\anacm,\abind)$ starting
from configuration $\acfg$ on database instance $\aninst$ is a sequence
of configurations and accesses 
\[\acfg_1, (\anacm_1, \abind_1), \ldots, (\anacm_{n-1},\abind_{n-1}), \acfg_n
\]
where $\acfg_1=\acfg$,
$(\anacm_1, \abind_1)=(\anacm,\abind)$, and $\acfg_{i+1}$ is a successor
configuration for
access $(\anacm_i,\abind_i)$ on $\acfg_i$.

By ``the certain answers to $Q$ after $p$'' we
mean the certain answers  to $Q$ on
$\acfg_n$, where $p$ terminates in configuration $\acfg_n$.

Given a path $p$, the  \emph{truncated path} of $p$ is the maximal
subpath 
\[
\acfg_1, (\anacm_2, \abind_2), \acfg'_2, \ldots , (\anacm_i, \abind_i),
\acfg'_i
\]
such that  each $(\anacm_j, \abind_j): 2 \leq j \leq i$ is a well-formed
access at $\acfg'_{j-1}$ (with $\acfg'_1=\acfg_1)$. That is, we eliminate 
the initial access in $p$, and then find
 the longest subpath of $p$ that did not depend on
this initial access.

We say that an access $(\anacm,\abind)$ is {\it long-term relevant ($\lti$) for $Q$} at
configuration $\acfg$ if
for some instance $\aninst$ consistent with $\acfg$, for some path
$p$ beginning with $(\anacm,\abind)$ the certain answers to $Q$ at $p$ are
different from those at the truncated path of $p$.

\begin{example}
Suppose that we have a schema
with relations $S, T$, 
and a query $Q= S \bowtie T$. Suppose we have a configuration $\acfg$ in
which $S$ and $T$ have not yet been accessed, 
and there is a dependent access method on~$T$.
Now consider an access ($\anacm, \abind$) on $S$. It is  long-term
relevant for $Q$, since it is possible that  $(\anacm, \abind)$
returns some new values, and using these values to access $T$
could yield some new tuples for $Q$.
\end{example}

When we speak about the problem of ``determining
whether an access is relevant'', we mean either of the problems
$\ii$ or $\lti$.

\paragraph*{The complexity of relevance}

We make a few general observations about the complexity of
determining if an access is relevant for a query.

First, we observe that there is a tight relation between the general
question of relevance and the special case of Boolean queries.
For a number $k$, let $\ii$($k$) be the problem of determining whether
an access in a given configuration is immediately relevant for a query with output arity $k$, relative
to a schema, and similarly for $\lti$.
Let $\vec c_k$ be a tuple of $k$ new constant symbols.
For any fixed $k$ we can solve $\ii$($k$) by considering every tuple of items that
come either from the
configuration or from $\vec c_k$ substituting them in for the head of the query
and then determining whether the access is $\ii$ on the configuration for the Boolean query thus created.
This shows:
\begin{proposition} \label{prop:reducebool}
Let $k$ be any number. There is a polynomial time reduction
from $\ii(k)$ to $\ii(0)$, and from $\lti(k)$ to $\lti(0)$.
\end{proposition}

We will thus focus on the Boolean case $k=0$ in this work.

Second, note that checking that an access is relevant, for any of the
notion of relevance we have defined, requires that we know that the
query
is not already satisfied in the configuration, which is
\complexity{coNP}-hard.

\section{Relevance and Containment}
\label{sec:containment}

In this section, we introduce the notion of
\emph{containment of queries under access limitations} and show how it is
strongly related to long-term relevance. We will use this connection to
ascertain the complexity of relevance in some cases.

\paragraph*{Containment under access limitations}
Query containment under access limitations was
shown to be decidable by Li and Chang~\cite{li2001answering-icdt},
and further investigated by Calì and Martinenghi
in~\cite{cali2008conjunctive}. We adapt here the definition to our
setting, and show further in Proposition~\ref{prop:containment-cali} that the
definition of~\cite{cali2008conjunctive} is essentially a special case of
ours. We give the definition for queries of arbitrary arity, but as we
explained we will focus on the Boolean case further on.

\begin{definition}
Let $\sch$ be a schema and $\acs$ a set of access methods over $\sch$.
Let $Q_1$ and $Q_2$ be two queries defined over $\acs$ and $\acfg$ a
configuration over $\sch$. We assume $Q_1$ and $Q_2$ have the same arity.
We say that \emph{$Q_1$ is contained in $Q_2$ under $\acs$ starting
from $\acfg$}, denoted $Q_1\sqsubseteq_{\acs,\acfg} Q_2$, if for every
configuration $\acfg'$ reachable from $\acfg$, $Q_1(\acfg')\subseteq
Q_2(\acfg')$. We simply say that \emph{$Q_1$ is contained in $Q_2$ under
$\acs$} and write $Q_1\sqsubseteq_\acs Q_2$ if $\acfg$ is the
empty configuration.
\end{definition}

As noted in~\cite{cali2008conjunctive}, in the presence of dependent
accesses, the notion of query containment under access limitations is
strictly weaker than the usual notion of query containment:

\begin{example}
Let $R$ and $S$ be two unary relations with the same domain,
each one with a single dependent
access method:
Boolean for $R$, and free for $S$. Consider queries $Q_1=\exists x\;R(x)$
and $Q_2=\exists x\;S(x)$. Starting from the empty
configuration, the only well-formed access paths that make $Q_1$ true,
i.e., produce an $R(x)$ atom, must first access $S$ and produce $S(x)$.
This means that $Q_1\sqsubseteq_\acs Q_2$ while, obviously,
$Q_1\not\sqsubseteq Q_2$.
\end{example}

More generally, many classical results that hold for
the classical notion of query containment are not true any more in the
presence of access constraints: for instance,
query containment of conjunctive queries 
cannot be tested by the existence of a homomorphism, and
query containment of
unions of conjunctive queries does not mean that all disjuncts of the
first query are contained in some disjunct of the second query, as is
true without access constraints~\cite{SagivY80}.

\paragraph*{Relating containment to relevance}
Query containment under access limitations is of interest in its
own right, but also for the connection
to
long-term relevance. We begin by showing that containment under access limitations
can be reduced to the complement of long-term relevance.

\begin{proposition}
\label{prop:containment-to-relevance}
There is a polynomial-time many-one reduction from the problem of query
containment of Boolean CQs (resp., PQs) under access limitations, starting from a
given configuration, to determining whether an access is \emph{not}
long-term relevant to
a Boolean CQ (resp., PQ), in another configuration. If the query is
a PQ, the configuration can be chosen to be the same.
\end{proposition}

\begin{proof}
For positive queries, the proof works by ``coding two queries
as one disjunction'' -- we create a query $\tau(Q,Q')$ and access
such that  if the access returns
successfully, then  the query is equivalent to $Q$, and otherwise
to $Q'$.
Disjunction can be eliminated by the idea of ``coding Boolean
operations in relations'', which will be used often in this paper.

We first show the result for PQs and then extend the result to
CQs.
Let $\sch$ be a schema, $\acs$ a set of access methods, $\acfg$ a
configuration over $\sch$, and $Q_1$, $Q_2$ two PQs over $\sch$. We
extend $\sch$ and $\acs$ into $\sch'$, $\acs'$ by adding a fresh unary
relation~$A$, with a Boolean access method. Let $c$ be a fresh constant.
We set $Q'=
\left(\left(\exists
x\,A(x)\right)\lor Q_2\right)\land Q_1$. We claim the following:
\[
Q_1\sqsubseteq_{\acs,\acfg} Q_2
\iff
\text{$A(c)?$ is not $\lti$ for $Q'$ in $\acfg$.}
\]

Assume $Q_1\not\sqsubseteq_{\acs,\acfg} Q_2$. Then there exists a
configuration $\acfg'$, reachable from $\acfg$, such that
$Q_1(\acfg')$ is true and $Q_2(\acfg')$ is false.
Consider a path $p$ leading from
$\acfg$ to $\acfg'$. We extend $p$ into a path $p'$ by adding, as first
access, $A(c)?$ that returns true. Then, $p'$
leads from $\acfg$ to $\acfg''=\acfg'\cup\{A(c)\}$.
Since $\acfg$
does not contain any fact about $A$, $\exists x\,A(x)$ is false in
$\acfg'$ and thus $Q'(\acfg')$ is false
while $Q'(\acfg'')$ is true. This means
$p'$, whose truncation is $p$,
is a witness that $A(?)$ is $\lti$ for $Q'$.

Conversely, assume that $A(c)?$ is $\lti$ to $Q'$. There exists a path
$p'$ which starts with an access $A(c)?$ and reaches from $\acfg$ a
configuration $\acfg''$ such that the truncation $p$ of $p'$ reaches from
$\acfg$ a configuration $\acfg'$ with
$Q'(\acfg'')$ true and $Q'(\acfg)$ false. Since $\acfg$ and $\acfg'$ are
different, it means that the response to $A(c)?$ was true. Without loss
of generality, we can assume that no other access in $p'$ made use of the
relation $A$, since no information gained this way would change the
result of $Q'$ once $A(c)$ is known. This means $p$ is a well-formed
access path starting from $\acfg$ for $\acs$, leading to $\acfg'$ and
$\acfg''=\acfg'\cup\{A(c)\}$.
Since $\acfg'$ has no fact about $A$ and the same facts as $\acfg''$
about the relations in $Q_1$, $Q_1(\acfg')$ is true and
$Q_2(\acfg')$ is false, i.e., $Q_1\not\sqsubseteq_{\acs,\acfg} Q_2$.

\newcommand{\B}{\mathcal{B}}
Consider now the case when $Q_1$ and $Q_2$ are conjunctive; we have to
modify the query $Q'$ to remove the disjunction. We extend
$\acs'$ by adding an extra place to all relations except
$A$, with a new domain $\B$, and adding two extra relations: a
binary relation $\mathit{Or}$, with both attributes typed with~$\B$, and a unary
relation $P$, with attribute typed with $\B$. There are no accesses to
either $\mathit{Or}$ or $P$, and the configuration is extended with the following
facts: $\mathit{Or}(1,0)$, $\mathit{Or}(0,1)$, $\mathit{Or}(1,1)$, $P(1)$. We extend $Q_1$ and $Q_2$
into $Q_1''(b)$ and $Q_2''(b)$ by simply putting the same variable $b$ in
each additional place. We choose the domain of the unique attribute of
$A$ to also be $\B$ and add $A(0)$ to the configuration.
Finally, we add to the configuration a fact
$R(c_1\dots c_k,0)$ for each relation~$R$, where the $c_i$'s are
arbitrary constants of the appropriate domain (but the same constant is
used for all attribute of the same domain, across relations),
and replace each ground
fact $R(d_1\dots d_k)$ of the configuration with $R(d_1\dots
d_k,1)$. Let $\acfg_{2}$ be
the new configuration.
We construct a conjunctive query
$Q''$ as:
\[
Q''=\exists
b_1\exists b_2\exists b\;A(b_1)\land Q_2''(b_2)\land \mathit{Or}(b_1,b_2)\land
Q_1''(b)\land P(b).
\]

Then $A(1)?$ is $\lti$ for $Q'$ in $\acfg$ if and only if $A(1)?$ is
$\lti$ for $Q''$ in $\acfg_2$. Assume the former: there is a path
$p$ that witnesses long-term relevance of $Q'$. We modify $p$ as $p_2$ 
by putting the
constant $1$ into the extra place of each relation. Then $Q''$ is true in
$\acfg_2+p_2$ and $Q''$ is false in $\acfg_2$ plus the truncation of $p_2$.
Conversely, the atom $P(b)$ enforces that witnesses of long-term
relevance of~$Q''$ can be turned into witnesses of long-term relevance
of~$Q'$.
\end{proof}

We can thus prove lower bounds for relevance using lower bounds for
containment. As an example, query containment under access limitations
obviously covers the classical notion of query containment (just make all
access methods free). This immediately entails that long-term relevance
is \conp-hard for CQs and \sigmatp-hard for PQs, even if all variables
are from the same abstract domain.
Conjunctive query
containment in the presence of datatype restrictions and fixed relations
is $\pitp$-hard
(this follows from \cite{SagivY80}) and hence containment under access in
our setting is $\pitp$-hard.
 We will show
that this latter lower bound actually already holds for conjunctive
queries even in very restricted settings  (cf.\
Proposition~\ref{prop:ltisoundind}).

In the other direction, from relevance to containment, we also
have a polynomial-time many-one reduction, but only for positive queries
and only for Boolean accesses.

\begin{proposition} \label{prop:reltoconpos}
There is a polynomial-time many-one reduction from the problem of
long-term relevance of a Boolean access for
a Boolean positive query in a
given configuration,
to the \emph{complement} of query
containment of Boolean positive queries under access limitations,
starting from another configuration.
\end{proposition}

\begin{proof}
We assume given a schema $\sch$ and a set of access methods $\acs$. Let
$\acfg$ and $Q$ be, respectively, a configuration and a positive query
over $\sch$. We consider an access $(\anacm, \abind)$ with
$\anacm\in\acs$. Let $R=\sof(\anacm)$. To simplify the presentation we
assume that input attributes of $\anacm$ come before output attributes.

\newcommand{\isbind}{\relation{IsBind}}
We add to $\sch$ a relation $\isbind$ with the same arity $k$ and
variable domains as $\abind$, without any access. We add to $\acfg$ the single fact
$\isbind(\abind)$ and denote the new configuration by $\acfg'$. We rewrite $Q$
as~$Q'$ by replacing every occurrence of $R(i_1\dots i_k,o_1\dots o_p)$
with \[R(i_1\dots i_k,o_1\dots o_p)\lor \isbind(i_1\dots i_k).\]

Then $(\anacm,\abind)$ is $\lti$ for $Q$ in $\acfg$ if and only if
$Q'\not\sqsubseteq_{\acs,\acfg'} Q$.

Assume $(\anacm,\abind)$ is $\lti$ for $Q$ in $\acfg$. There is a
well-formed path $p$ starting with the access, with truncated path~$p'$,
such that $Q$ is true in $\acfg+p$ and false in $\acfg+p'$ (and, since
$\isbind$ does not occur in $Q$, also false in $\acfg'+p'$).
For every
subgoal $R(i_1\dots i_k,o_1\dots o_p)$
of $Q$ that is witnessed by the first access of $p$, $\isbind(i_1\dots
i_k)$ is true in $\acfg'$ and thus a witness that $Q$ is true in $\acfg+p$
yields a witness that $Q'$ is true in $\acfg'+p'$.

Conversely, assume there is a path $p'$ such that $Q'$ is true and $Q$ is
false in $\acfg'+p'$. For every $R(i_1\dots i_k, o_1\dots o_p)$ that
is false in $Q$ while the corresponding disjunction is true in $Q'$, we 
construct a ground fact $R(\abind, c_1\dots c_p)$ where $(c_1\dots
c_p)$ are the constants that $(o_1\dots o_p)$ are mapped to in a witness
of $Q'$. Then we build a new path $p$ by prepending to $p'$ all these
ground facts, returned by $(\anacm,\abind)$. The path~$p$ witnesses that
$(\anacm,\abind)$ is $\lti$ for $Q$ in $\acfg$.
\end{proof}

Finally,
for conjunctive queries, we prove similarly
a different form of reduction, a Turing
reduction in nondeterministic polynomial time:
\begin{proposition}
\label{prop:relevance-to-containment-conjunctive}
Long-term relevance of a Boolean access for a CQ can be decided
with a nondeterministic polynomial-time algorithm with access to an
oracle for query containment of CQs under access limitations.
\end{proposition}

\begin{proof}
As in the previous reduction, we assume we are given $\sch$, $\acs$, $\acfg$, a
conjunctive query $Q$, and an access $(\anacm, \abind)$.

We pose $Q=Q_1\land Q_2$ with $Q_1$ all subgoals of $Q$ which are
compatible with the access (same relation, and no mismatch of constants
with the binding). We guess, nondeterministically, a subset $Q'_1$ of
$Q_1$ distinct from $Q_1$.
We use the containment oracle to test whether $Q'_1\land
Q_2\sqsubseteq_{\acs,\acfg} Q$. If not, we return true. If all guesses
lead to a $Q'_1\land
Q_2$ contained in $Q$ we return false.

Let us show that this algorithm tests long-term relevance.

Assume
$(\anacm,\abind)$ is long-term relevant. Then there exists a path $p$
starting with the access and with truncation $p'$
such that $Q$ is true in $\acfg+p$ and false in $\acfg+p'$. The only
difference between $p$ and $p'$ being the tuples returned by the access,
$Q_2$ is true in $\acfg+p'$. Since $Q$ is true in $\acfg+p$, some
non-empty subset of $Q_1$ is witnessed by the access. Let $Q'_1$ be the
complement of this subset, that is, the facts of $Q_1$ that were
not witnessed by the access.
Then $Q'_1\land Q_2$ is true in $\acfg+p'$ while
$Q$ is false and the algorithm returns true.

Assume now that the algorithm returns true. Then there exists a subset $Q'_1$
of $Q_1$ and a path $p'$ such that $Q'_1\land Q_2$ is true and $Q$
is false in $\acfg+p'$. Let $Q''_1$ be the complement of $Q'_1$ in $Q_1$.
Since $Q'_1\land Q_2$ is true  in $\acfg+p'$, there is a 
homomorphism from $Q'_1\land Q_2$ into $\acfg+p'$.
We extend this into a homomorphism from $Q$ to an extension $\acfg'$
of $\acfg+p'$
by mapping
remaining variables to some fresh constants in the appropriate domains.
Observe that the facts that were added to $\acfg'$ can all be produced by
the access. We thus have a witness for long-term relevance.
\end{proof}

Both reductions will be used to show that upper bound results 
can be lifted from containment to relevance. Even though the latter
reduction seems weak, it will be enough for our purpose (see
Section~\ref{sec:dep}).

\paragraph*{Containment and containment}
Our notion of query containment under access limitations starting from a
given configuration differs in some ways from the notion introduced by Calì
and Martinenghi in~\cite{cali2008conjunctive}. In this part of the paper,
to emphasize the distinction, we refer to the former as
\emph{config-containment} and to the latter as \emph{CM-containment}.
The differences are as follows:
\begin{inparaenum}[(i)]
\item In CM-containment, there is always exactly one access method per
regular relation, whereas in config-containment there may
be zero or several access methods per
relations;
\item CM-containment is defined with respect to a set of existing
constants (of the various abstract domains) that can be used in access
paths, while config-containment, as its name implies,
uses a more general notion of pre-existing configuration,
with constants as well as ground facts;
\item Access methods in CM-containment are always exact, they return the
complete collection of facts compatible with the binding that are present
in the database instance; we do not make such an assumption for
config-containment and merely assume accesses are sound;
\item In addition to regular relations,
CM-containment also supports \emph{artificial relations}
which are unary (monadic) relations ``whose content is accessible and
amounts only to'' some constant value. Since they are not described in the
definition of the CM-containment, but are added for the purpose of
eliminating constants in the queries (see p.~331
of~\cite{cali2008conjunctive}), it is not exactly clear what the
restrictions are on these artificial relations.
To the best of our understanding,
they correspond in our setting
to relations without any access methods, except that
config-containment allows them to have arbitrary arity. It would be
interesting to see if the result of this paper (and
of~\cite{cali2008conjunctive}) can be extended to the case where all
relations have an access method.
\end{inparaenum}

Among these four, the significant difference is the forbidding of
multiple accesses per relation, which means CM-containment is a special case
of config-containment:

\begin{proposition}
\label{prop:containment-cali}
There are polynomial-time reductions in both directions between
CM-containment and the special case of
config-containment
when relations have at most one access method and relations without an
access method have arity bounded by a constant $K$. The query language (CQs,
PQs) is preserved by the reductions.
\end{proposition}

\begin{proof}
The argument from CM-containment to config-containment is simple,
since config-containment allows a richer initial condition.
The reduction the other way requires us to code the configuration
in the contained query.

We first observe that when there is only one access method per relations,
and there are no pre-existing ground facts, an access method is exact if
and only if it is \emph{idempotent}, i.e., it always returns the same
result for a given binding: as there is no other access methods to give
another view of the same relation, the instance can just as well be
assumed to be exactly what is returned by the access method. To show the
reduction from CM-containment to config-containment, we then just need to
show that idempotence does not have an impact on query containment, which
is straightforward: if there is a path that witnesses non-containment,
then this path can always be assumed to use each combination of access
method and binding only once (we may regroup all such accesses as one),
which leaves the path well-formed and respects idempotence of accesses.

Conversely, let us now explain the reduction from config-containment to
CM-containment, in the case when there is at most one access method per
relation and a bound on the arity of the relations without access methods.
We have thus schema $\sch$, access methods $\acs$, 
queries $Q_1$ and $Q_2$, configuration $\acfg$. The main idea is to
encode the configuration into the contained query: we let $C$ be a big
conjunction consisting of all ground facts from the configuration. We first
explain the reduction assuming it is possible to have relations
without access methods of arbitrary arity in
CM-containment; we discuss afterwards how to use get rid of them.
Then we claim
that $Q_1\sqsubseteq_{\acs,\acfg} Q_2$ if and only if $Q_1\land C$ is
CM-contained in $Q_2$ with respect to $\acs$ and the set of constants
$\adom(\acfg)$.

To see this, assume $Q_1\not\sqsubseteq_{\acs,\acfg} Q_2$. Then there
exists a well-formed path $p$ such that $Q_1$ is true and $Q_2$ is false
in $\acfg+p$. Since $Q_1$ is true in $\acfg+p$, and $C$ is true in
$\acfg$, $Q_1\land C$ is true in $\acfg+p$, and $p$ is a witness of
CM-containment of $Q_1\land C$ in $Q_2$ (again, we merge repeated
identical accesses to ensure idempotence). Assume now $Q_1\land C$ is not
CM-contained in $Q_2$ under $\acs'$ with respect to $\adom(\acfg)$. Then
there is a path $p$ starting with constants of $\adom(\acfg)$ (and facts
of the relations without accesses) that makes
$Q_1\land C$ true and $Q_2$ false. We transform $p$ into another path
$p'$ by removing every fact that is in $\acfg$. Since $\acfg+p'$ contains
exactly the facts that were in $p$ (since $C$ is true, all facts of
$\acfg$ were present in~$p$), $Q_1$ is true and $Q_2$ is false in
$\acfg+p'$.

We still need to show that CM-containment of $Q_1$
into $Q_2$ with relations without access methods can be reduced to the
case when the only relations without access methods are monadic. Let $R$
be a relation of arity $1<k<K$ with no access methods.
We introduce a new monadic
relation $R_a$ without access methods for each attribute $a$ of $R$. This
relation contains the projection of $R$ along its attribute $a$. We then
conjunctively add to $Q_1$ all ground facts known about relation $R$. At
the same time, we add disjunctively to $Q_2$ all ground facts formed with
constants of the $R_a$ that are
known \emph{not to be true} about relation $R$ (when $K$ is a constant,
their number is polynomial). Let $Q_1'$ and $Q_2'$ be the
resulting queries. Then $Q_1$ is CM-contained in $Q_2$ (with respect to
constants and ground facts of relations without accesses) if and only if
$Q_1'$ is CM-contained in $Q_2'$ (with respect to constants and ground
facts of \emph{monadic} relations without accesses).

Now, $Q_2'$ adds disjunction, so the reduction works for positive
queries, but needs to be adapted still for conjunctive queries. We use
the same trick as in the proofs of
Proposition~\ref{prop:containment-to-relevance} and
Theorem~\ref{thm:conj-nexptime-hard},
that we briefly sketch here: we add a fixed relation $\mathit{Or}$ that
contains the truth value of the \emph{or} operator. This relation will
undergo the same processing as the other non-monadic fixed relations.
Then we add an extra Boolean place to every relation.
We replace disjunction with conjunctions, adding extra $\mathit{Or}$
predicates to compute the disjunction of
the result of the extra Boolean place of every relation in $Q_2'$ while
we keep $Q_1'$. We add facts to the configuration (that is, to $Q_1$)
with the extra place set to $0$ so
that the individual elements of $Q_2'$ still match whereas we require the
extra place in $Q_1$ to match to $1$.
\end{proof}

As can be verified, all hardness proofs for containment that we present
in this paper make use of no relation with multiple access methods, and the
arity the relations with no accesses is bounded by $3$. This means all
lower bounds for config-containment obtained in this paper 
yield identical lower bounds for
CM-containment.

Calì and Martinenghi give in~\cite{cali2008conjunctive} a proof that
the CM-containment problem is decidable in \conexptime{} for conjunctive
queries,
without any lower bound. We give in Section~\ref{sec:dep} the same upper
bound of \conexptime{} for config-containment, as well as a matching
lower bound.

By having the possibility of several access methods per relation, and of some
fixed base knowledge given by relations without accesses, we allow 
modeling of a more realistic setting, where multiple sources of the deep Web
may share the same schema, and where we want to ask queries over these
sources as well as over some fully accessible local knowledge.

\section{Independent Accesses} \label{sec:indep}

We establish in this section complexity bounds for the problem of
determining whether an access is relevant to a query, when all access
methods are \emph{independent}.
Our upper bounds will be fairly immediate -- the main work involved
is in the lower bounds.

In the case of independent accesses, we have some immediate
facts:
\begin{inparaenum}[(i)]
\item An access to a relation
that is  not mentioned in the query can not be relevant in either of our senses.
\item A path witnessing the fact that an   access is long-term relevant can always be pruned to
include only subgoals of the query, with each subgoal occurring at most once.
This gives a bound of \sigmatp{} in combined complexity
for checking long-term relevance, since checking
that the truncation of the  path  does not satisfy the query is in
$\conp$ for the considered query languages.
\item Since constants can be guessed at will, abstract domain
constraints do not have any impact on relevance for independent accesses
and we can assume all attributes to share the same domain.
\end{inparaenum}

We study the complexity of whether an access is relevant, for immediate
and long-term relevance.

\paragraph*{Immediate relevance}

The following result characterizes the combined and data complexity of $\ii$
for independent access methods.

\begin{proposition} \label{prop:iisoundind}
We assume all access methods to be independent.
Given a Boolean positive query $Q$, configuration $\acfg$,
and access $(\anacm,\abind)$, determining whether
$(\anacm,\abind)$ is immediately
relevant for the query $Q$ in $\acfg$ is a \complexity{DP}-complete
problem. If we know the query is not certain in $\acfg$, then the problem is
\complexity{NP}-complete. Lower bounds hold even if the query is
conjunctive.

If the query is fixed, the problem is in \acz.
\end{proposition}

\begin{proof}
Let us first give some intuition of the proof.
Membership in \complexity{DP} works via guessing
a witness configuration and verifying it. One can show that
the witness need not be large, and verifying it requires
checking a conjunctive query and a negation of a conjunctive
query.
Hardness uses a coding of satisfiable/unsatisfiable pairs
of queries.

We now show that the problem is in \complexity{DP}.
We describe an \complexity{NP}
algorithm for checking that the access is $\ii$, provided that $Q$ is false
in $\acfg$. This, together with the observation that if $Q$
is certain in $\acfg$ (which can be tested in 
\complexity{NP}), the access cannot be $\ii$, gives membership in
\complexity{DP}.

Let $c$ be a fresh constant.
We guess a mapping $h$ from the variables of the query to
$\adom(\acfg) \cup \{c\}$. For a subgoal $G$ of $Q$, let $h(G)$ be the
ground fact obtained
by replacing any variable $x$ in $G$ by $h(x)$. We construct from $Q$ a
positive Boolean expression $\phi$ obtained by replacing each subgoal
$G=R(x_1,\dots,x_n)$ with:
\begin{itemize}
\item \emph{true} if $h(G)\in\acfg(R)$, or if $R=\sof(\anacm)$ and places
from $\inat(\anacm)$ in $G$ are mapped to $\abind$;
\item \emph{false} otherwise.
\end{itemize}
We check that $\phi$ evaluates to \emph{true}.
If there exists such an $h$ return true. Otherwise, return false.

Let us show that this algorithm checks that the access is $\ii$, provided
that $Q$ is false in $\acfg$. 
Suppose the algorithm returns true, and fix a witness $h$.
We extend $\acfg$ by adding $h(G)$ for every subgoal
where
the accessed relation is $R$ and the input places match the binding.
This configuration witnesses that the access
is immediately relevant.
In the other direction, if $h$ is $\ii$, there is
a set of tuples $C$ matching the schema of $R$ and the binding $\abind$
such that the query is satisfied in the extension formed
from $\acfg$ by adding $C$ to $R$. Clearly if all the values in $C$
other than those in $\adom(\acfg)$ are identified, then
$Q$ is still satisfied; hence we can assume there is only one
such value. One can check that $h$ satisfies the condition
in the algorithm, hence the algorithm returns true.

The algorithm gives a way to test immediate relevance
by evaluating 
some first-order query $Q'$ over the
configuration. For example, if $Q=\exists x\exists y\; R(x,y)\land
S(x)\land S(y)\land T(y)$ and the access is $S(0)?$ then
\begin{align*}
Q'=&\lnot Q\land \big((\exists y\; R(0,y)\land S(y)\land T(y))\lor{}\\
&(\exists x\;
R(x,0)\land S(x)\land T(0))\lor (R(0,0)\land T(0)\big)
\end{align*}
($S(x)$ is used as a shortcut notation for $\acfg(S)(x)$).
This query is potentially exponential in the size of $Q$, but for a fixed
query $Q$, this shows that immediate relevance is \acz.

We now prove \complexity{DP}-hardness. Let $Q_1$, $Q_2$ be two Boolean
conjunctive queries over disjoint schemas $\sch_1$ and $\sch_2$
and $\aninst_1$, $\aninst_2$
two database instances over, respectively, $\sch_1$ and $\sch_2$.
We reduce from the problem of determining whether $Q_1$ is not true in
$\aninst_1$ and $Q_2$ is true in $\aninst_2$. Let $\sch'_1$ and
$\sch'_2$ be the modifications of $\sch_1$ and $\sch_2$ where every
relation has an extra attribute. We take as $\sch$ the union of $\sch'_1$
and $\sch'_2$, together with an extra unary relation $R$. We assume that
the only access method available in $\sch$ is a Boolean access
method on $R$. Let $a$ and $b$
be two fresh constants. We construct a configuration $\acfg$ as follows:
\begin{itemize}
\item for each ground fact $S_1(c_1,\dots,c_n)\in I_1$, we add the tuple
$(c_1,\dots,c_n,a)$ to $\acfg(S_2)$;
\item for each ground fact $S_2(c_1,\dots,c_n)\in I_2$, we add the tuple
$(c_1,\dots,c_n,b)$ to $\acfg(S_2)$;
\item for each $k$-ary relation $S_1\in\sch_1$, we add a $k+1$-ary tuple 
$(b,\dots,b)$ to $\acfg(S_1)$;
\item for each $k$-ary relation $S_2\in\sch_2$, we add a $k+1$-ary tuple 
$(a,\dots,a)$ to $\acfg(S_2)$;
\item we add $a$ to $\acfg(R)$.
\end{itemize}
Consider the query $Q_1'(x)$ obtained from $Q_1$ 
by adding an extra variable x to each
subgoal, and similarly for $Q_2'(x)$. Then we consider the query
$Q=\exists x\; Q_1'(x)\land Q_2'(x)\land R(x)$. We claim that
$R(b)?$ is $\ii$ for $Q$ in $\acfg$ if and only if $Q_1$ is not true in
$\aninst_1$ and $Q_2$ is true in $\aninst_2$.

Clearly $R(b)?$ is $\ii$ for $Q$ in $\acfg$ if and only if $Q$ is not
certain
in $\acfg$ (but since we know that $Q_2'(a)\land R(a)$ is true in
$\acfg$, $Q$ is certain if and only if $Q_1'(a)$ is certain, i.e., $Q_1$ is
true in $\aninst_1$) and $Q$ is certain in $\acfg \cup \{R(b)\}$.
Since $Q_1'(b)$ is certain in $\acfg$,
this condition is just that $Q_2'(b)$ is certain in $\acfg + C$, i.e., $Q_2$
is true in $\aninst_2$.

In the case we know the query is not certain in $\acfg$, we obtain
\complexity{NP}-hardness by considering only $Q_2$.
\end{proof}

\paragraph*{Long-term relevance}

We now move to $\lti$ for independent methods.

\begin{example}
Consider $Q=R(x,5) \wedge S(5,z)$,  a configuration
in which only $R(3,5)$ holds, and an access method 
to $R$ on the second component. 
Clearly, an access $R(?,5)$  is not long-term relevant, since
any witness $x$ discovered in the response to this
access  could be replaced by $3$. On the other hand, if the configuration
had $R(3,6)$ then an access  $R(?,5)$ would be long-term relevant.
\end{example}

Let us first consider a simple case: that of conjunctive queries where
the accessed relation only occurs
once. In this particular case, it is possible to decide $\lti$
by checking whether the subgoal containing the accessed relation is not in a
connected component of the query that is already certain.

More formally, 
let $R=\sof(\anacm)$. Let $h$ be the (necessarily unique) partial mapping
substituting the binding for the corresponding elements of the subgoal of
the conjunctive query
$Q$ containing $R$, and $Q_h$ be the query obtained by applying $h$ to
the variables of~$Q$. If no such $h$ exists (since the subgoal conflicts
with the binding), then clearly the access is not $\lti$. Otherwise,
let $G(Q_h)$
be the graph whose vertices are the subgoals of~$Q_h$ with an edge
between subgoals if and only if they share a variable in $Q_h$. Let
$Q_h-\mathit{Sat}(\acfg)$ be the query obtained from $Q_h$ by removing
any subgoal that lies in a component of $G(Q_h)$ that is satisfied
in $\acfg$. If $Q_h-\mathit{Sat}(\acfg)$ contains the $R$
subgoal, we return true, otherwise we return false.

If this algorithm returns true, then let $g$ be the subgoal
containing $R$, $h$ the homomorphism, and $C$ be the component of
$g$ in $G(Q_h)$. If $C$ contains only $g$,
 then
clearly the access is long-term relevant, since we can consider any
path that begins with the access and then continues through every subgoal.
If $C$  contains other subgoals, then there is some variable shared between
$g$ and other subgoals; again we take a path beginning at $R$, accessing
each additional subgoal in turn (in an arbitrary order) and using
new elements as inputs while returning  elements not in $\acfg$
for all variables in the subgoals. Note that we do not have to access the other subgoals using
the shared variables -- we use an arbitrary access method
for other relations (we know at least one exists), and choose
a response such that the  results match the subgoals. 
This witnesses that the access is $\lti$.
Conversely, suppose $R$ does not occur in $Q_h-\mathit{Sat}(\acfg)$ and
that we have a path
$p$
witnessing that the access is long-term relevant. Let $h'$ be a homomorphism
from $Q$ into the $\acfg+p$, and let $h''$ be formed from $h'$ by replacing all
elements in $C$ by witnesses in $\acfg$. The existence of $h''$ proves that
the path is not a witness of the fact the access is $\lti$.

We have thus the following complexity result in this particular case of
conjunctive queries with only one occurrence of the accessed relation
(hardness is again shown via a coding argument).

\begin{proposition} \label{prop:oneaccess}
We assume all access methods to be independent.
Given a Boolean conjunctive query $Q$, a configuration $\acfg$,
and an access $(\anacm,\abind)$, such that $\sof(\anacm)$ only occurs
once in $Q$, determining whether
$(\anacm,\abind)$ is long-term
relevant for the query $Q$ in $\acfg$ is a \complexity{coNP}-complete
problem.
\end{proposition}

\begin{proof}
The upper bound has been shown in the body of the paper. We know that the problem is
\conp-hard because we need to test whether $Q$ is not certain in $\acfg$,
but let us show that this \conp-hardness holds even when we know this to
be true.
We reduce from
the complement of conjunctive query satisfaction.
Let $Q$ be a Boolean conjunctive query over schema $\sch$, $\aninst$ a
database instance over $\sch$,
and $c,d$ be  new constants. Let schema $S'$ modify the original schema
$S$ by adding an additional place to each relation and including the additional
unary relations $R_1,R_2$. We assume $R_1$ and $R_2$ both have
access methods requiring an input. Let  configuration
$\acfg'$ for $S'$ be formed by adding the constant $c$ in the additional
place for every tuple seen within every relation of $\aninst$,  letting $R_1$ 
have   $c$ known to be true, and nothing known about $R_2$. Finally, let
$Q'(z)$
be formed from $Q$ by replacing every atomic formula $L(\vec x)$ with
$L(\vec x, z)$, where $z$ is a new variable, and adding additional
subgoals $R_1(z),R_2(w)$, where $w$ is another new variable.
Notice that
$Q$ holds in $\aninst$ if and only if the subquery of $Q'(c)$ without the
subgoal $R_2(w)$ is certain in $\acfg'$.
Consider an access to $R_1$ with input $d$. We claim that this
access is long-term relevant for $Q'(z)$ if and only if $Q$ is not satisfied
in $\aninst$. In one direction, if $Q$ is not satisfied in $\aninst$, consider
a path that first accesses $R_1(d)$, returning true,  then accesses
$R_2(e)$ for some $e$ returning true, followed by any path realizing
$Q'(d)$ that does not access $R_1$ or $R_2$. This path is consistent,
and it witnesses that the access is long-term
relevant. If $Q$ is satisfied in $\aninst$, then any witness
for $Q'$ with $z=d$ can be replaced with one with $z=c$, hence
the access to $R_1(d)$ cannot be $\lti$.
\end{proof}

In the case where a relation is repeated, however, 
simply looking at satisfied components is not sufficient.

\begin{example}
Consider $Q=R(x,y) \wedge R(x,5)$, an empty
configuration, and an access
to $R$ with second component $3$ (i.e., $R(?,3)$). Clearly this access is
not long-term relevant in the empty configuration,
since in fact $Q$ is equivalent to the existential closure
of $R(x,5)$, and the access can reveal nothing about such
a query. But no subgoals are realized in the configuration.
\end{example}

In the general case (repeated relations, positive queries),
we can fall back on the \sigmatp{} algorithm described at the beginning
of the section.
Surprisingly, this is the best we can do
 even in very limited situations:
\begin{proposition} \label{prop:ltisoundind}
We assume all access methods to be independent.
Given a Boolean positive query $Q$, a configuration $\acfg$,
and an access $(\anacm,\abind)$, determining whether
$(\anacm,\abind)$ is long-term
relevant for the query $Q$ in $\acfg$ is \sigmatp{}-complete.
The lower bound holds even if $Q$ is conjunctive and known not to be certain
in~$\acfg$,
and even if all relations are freely accessible and all
variables have the same infinite datatype.

If the query is fixed, the problem is \acz.
\end{proposition}

\begin{proof}
The lower bound follows from results of Miklau and Suciu \cite{miklausuciu};
we explain the connection to their notion of \emph{criticality}, which
is closely related to relevance.
For a query $Q$ on a single relation $R$ and a finite domain $D$ (i.e., a
finite set of constants), a tuple $\vec t$ (with same arity as $R$)
is \emph{critical} for $Q$
if there exists an instance $I$ of $R$ with values in $D$ such
that deleting $\vec t$ from $I$ changes the value of $Q$.

\begin{quote}
\begin{theorem}[4.10 of \cite{miklausuciu}]
\label{thm:miklausuciu}
 The problem of deciding, for conjunctive query $Q$ and set $D$ whether a tuple $\vec t$
is not critical is
$\pitp$-hard, even for fixed $D$ and $\vec t$.
\end{theorem}
\end{quote}

It is easy to see that $\vec t$ is critical iff the 
Boolean access $R(\vec t)$ is $\lti$ in
the empty configuration (or, more precisely, in a configuration that only
contains constants of the queries but no facts for $R$):
this holds since $\lti$ is easily
seen to be equivalent to the existence of some
instance  of size at most $|Q|$ where adding the tuple
given by the access changes the truth value of $Q$ to true.
Hence the theorem above shows that $\lti$ is $\sigmatp$-hard even for a fixed
configuration.

\newcommand{\G}{\mathcal{G}}
The combined complexity upper bound has already been discussed. We now
discuss data complexity. We can assume without loss of generality $Q$ is
in disjunctive normal form (i.e., a union of conjunctive queries).
Let us present the \sigmatp{} algorithm slightly differently.
We make a guess for each subgoal $G$ of $Q$
of the following nondeterministic choices:
\begin{compactenum}
\item $G$ is not witnessed;
\item $G$ is witnessed by the configuration;
\item $G$ is witnessed by the first access;
\item $G$ is witnessed by a further access.
\end{compactenum}
For such a guess $h$, we write $\G^h_1$, $\G^h_2$, $\G^h_3$, $\G^h_4$ the
corresponding partition of the set of subgoals.
We restrict valid guesses to those where (i)~at least one of the disjuncts of
$Q$ has all its subgoals witnessed (i.e., in
$\G^h_2\cup\G^h_3\cup\G^h_4$);
(ii)~all subgoals in $\G^h_3$
are compatible
with $(\anacm,\abind)$ (if we want the access to also be part of the
input, we can easily encode this condition into the constructed formula).

\renewcommand{\H}{\mathcal{H}}
Let $\H$ be the set of all valid guesses (there is a fixed number of
them once the query is fixed). For a given $h\in\H$, we
rewrite $Q$ into two queries $Q'_h$ and $Q''_h$. $Q'_h$ is obtained from~$Q$
by replacing every variable mapped to an input attribute of subgoals
in $\G^3_h$ with the binding, by replacing every subgoal of
$\G^h_4$ with \emph{true}, and by dropping every disjunct of $\H$ that
has a subgoal in $\G^1_h$. $Q''_h$ is obtained from~$Q'_h$ by
replacing every subgoal in $\G^h_3$ with \emph{true}.
Then $(\anacm,\abind)$ is $\lti$ for $Q$ in $\acfg$ if and only if the
following first-order query evaluates to \emph{true} on $\acfg$:
\(
\bigvee_{h\in\H}Q''_h\land \lnot Q'_h.
\)
This query is exponential in the size of $Q$.
\end{proof}

\medskip

As shown, long-term relevance, even in
the independent case and for the relatively simple language of
conjunctive queries, is already at the second level of the polynomial
hierarchy in combined complexity. Introducing dependent accesses will
move the problem into the exponential hierarchy.

\section{Dependent Accesses} \label{sec:dep}

We now turn to the case when some of the accesses are dependent.
The results for $\ii$ are clearly the same
as in the independent case, since this only considers the impact of a single
access. We thus only discuss long-term relevance.

A naïve idea would be to show that a witness path is necessarily short,
by using the initial access to generate all constants needed to witness
the query. This, however, requires two things: the initial access needs to be
non-Boolean, and it should be possible for this access to generate
constants of all relevant domains.
This is clearly not a realistic assumption.

We deal \emph{only with  long-term relevance for Boolean accesses}.
We strongly rely for establishing the upper bound results of this section on the
connection between relevance and containment that was established in
Section~\ref{sec:containment}. Lower bound arguments will be based on
constraining paths to be exponentially or doubly exponentially long, using
reductions from tiling.

The upper bounds will make use
of methods due to Calì and Martinenghi,
which
are related
to those of Chaudhuri and Vardi \cite{cv94,cv97} for Datalog containment.
In \cite{cali2008conjunctive}, Calì and Martinenghi show that
we can assume
that counterexamples to containment of $Q$ in $Q'$ (or witnesses to long-term relevance)
are \emph{tree-like}. In the  containment setting,
 this means that every element outside the initial configuration and
the image of $Q$ under a homomorphism occurs in at most two accesses, one
as an output and possibly one as an input. Each element outside of
the configuration can be associated with an atom -- the atom that
generated that element as output.
For elements $n_1, n_2$ neither in $\acfg$ or $h(Q)$, say
that $n_1 \prec n_2$  if $n_2$  is generated by
an access to $n_1$, and let $\prec^*$ be the transitive
closure of $\prec$. Then the tree-like requirement corresponds
to the fact that $\prec^*$ is a tree.
This is exactly what  Calì and Martinenghi
refer to as a  \emph{crayfish} chase database, and 
in~\cite{cali2008conjunctive} they give the ``unfolding'' construction that
shows that such counterexample models always exist -- we use this often
throughout this section.
By exploiting the limited structure of a tree-like
database further we will be able
to  extend the upper bounds of \cite{cali2008conjunctive}, taking into account 
configuration constants and multiple accesses.

In contrast with
what happens for the independent case, results are radically different
for conjunctive and positive queries. We thus study the complexity of
long-term relevance and query containment in turn for both query
languages.

\subsection{Conjunctive Queries}

For conjunctive queries, we show we have \conexptime-completeness of
containment and \nexptime-completeness of the $\lti$ problem. We first
establish the hardness through a reduction of a tiling problem of an
exponential-size corridor that yields an exponential-size path.
Recall that the lower bound for query
containment directly implies the lower bound for relevance, thanks to
Proposition~\ref{prop:containment-to-relevance}.

\begin{theorem}
\label{thm:conj-nexptime-hard}
Boolean conjunctive query containment under access limitations
is \conexptime-hard. Consequently, long-term relevance of an access for a
conjunctive query is \nexptime-hard.
\end{theorem}
\begin{proof}
We show a reduction from the \nexptime-complete problem consisting in tiling a
$2^n\times 2^n$ corridor,
where $n$ is
given in unary, under horizontal and vertical
constraints~(see Section 3.2 of~\cite{johnson1990catalog}).
A tile will be represented by an access atom 
$\tile(t,\bb,\cc,x,y)$, where $\bb$ is the vertical position (i.e., the row), $\cc$ the horizontal position 
(i.e., the column), $t$ the   tile type, and $x,y$ are values that link one tile to 
the next generated tile in the witness path, as will become clear later. The $n$ bit binary representation of the decimal number $d$ is 
denoted by $[d]$.

For a given tiling problem with tile types $t_1,\ldots,t_k$, horizontal relation $H$, vertical relation $V$, and initial tiles of respective type $t_{i_0},\ldots t_{i_{m-1}}$,
we construct the following containment problem.
The schema has the following relations with their respective arities as superscripts: 
$\rboolean^1$, $\tiletype^1$, $\sametile^3$, $\horiz^3$, $\vertic^3$,
$\rand^3$, $\ror^3$, $\eq^3$, all of them having no access methods,
and  
$\tile^{2n+3}$, with a single access method whose input arguments are all
but the last.
We generate the following configuration $\acfg$:
\begin{quote}
\begin{tabbing}
$\rboolean(0), \rboolean(1);$\\
$\tiletype(t_1),\ldots,\tiletype(t_k);$\\
$\sametile(t_i,t_i,1)\text{ for }1\leq i\leq m,$\\
$\sametile(t_i,t_j,0)\text{ for }i\neq j,1\leq i,j\leq m;$\\
$\horiz(t_i,t_j,1)\text{ for all }\langle t_i,t_j\rangle\in H,$\\
$\horiz(t_i,t_j,0)\text{ for all }\langle
t_i,t_j\rangle\not\in H;$\\
$\vertic(t_i,t_j,1)\text{ for all }\langle t_i,t_j\rangle\in V,$\\
$\vertic(t_i,t_j,0)\text{ for all }\langle t_i,t_j\rangle\not\in V;$\\
$\rand(0,0,0), \rand(1,0,0),\  \rand(0,1,0),\rand(1,1,1);$\\
$\ror(0,0,0), \ror(0,1,1),\ror(1,0,1), \, \ror(1,1,1);$\\  
$\eq(0,0,1),\eq(1,0,0), \eq(0,1,0), \eq(1,1,1);$\\
$\tile(t_{i_0},[0],[0],c_0,c_1), \tile(t_{i_1},[0],[1],c_1,c_2).$
\end{tabbing}
\end{quote}
\newcommand{\B}{\mathcal{B}}
\newcommand{\T}{\mathcal{T}}
\newcommand{\C}{\mathcal{C}}
We use three domains: $\B$ (used for Booleans), $\T$ (used for tile
types), $\C$ (used for chaining up tiles). They are assigned
as follows: $\rboolean$, $\rand$, $\ror$, $\eq$ have all their
argument in $\B$; the argument of $\tiletype$ has domain
$\T$; $\sametile$, $\horiz$, and $\vertic$ have their first two arguments
of domain $\T$ and the third of domain $\B$; finally, the first argument
of $\tile$ has domain $\T$ and the remaining ones domain $\B$, except for
the last two, that are in $\C$.

We reduce to the complement of query containment of $Q_1$ into $Q_2$
where $Q_1$ is the atom \(\tile(u,[2^n-1],[2^n-1],x,y)\) and
$Q_2$ consists of the following conjunction of atoms:
\[\begin{split}
&\tile(u,\bb,\cc,x,y) \wedge \tile(v,\dd,\ee,y,z) \wedge
\tile(w,\fff,\ggg,y',z')\\&\quad \wedge\tile(q,\vv,\hh,y',z")\wedge
\BOOLCONS,\end{split}\]
where $u,v,w,q$ are variables intended for tile types, each of $\bb,
\cc,\dd,\ee,\fff, \ggg,\hh,\vv$ is a tuple of $n$ 
variables intended for Booleans, and $x,y,y',z,z',z"$ are variables intended for linking elements, and 
where $\BOOLCONS$ consists of a conjunction of  $\rand$, $\ror$, and $\req$ atoms 
imposing a number of Boolean constraints on the bit-vectors  $\bb,
\cc,\dd,\ee,\fff,\ggg, \hh, \vv$.
Since we want $\BOOLCONS$ to remain conjunctive, its construction is a
bit intricate, but in essence it states that there is ``something wrong''
with the tiling. More precisely, it consists of a conjunction of the following 
subformulas:

1.~A subformula $\SUB_1(i_1)$ such that $\SUB_1(0)$ holds true iff the functional dependency from the next-to-last argument 
of the $\tile$ relation to 
the $2n$ bit-valued attributes in the same relation is violated for some
tuple. To express this, we use the last two of the four $\tile$ atoms in 
the above formula, because they have both the same variable $y'$ in their next-to-last position. 
We require that some $f_i$ differs from the corresponding $v_i$ or some $g_j$ differs from its 
corresponding~$h_j$. We can assert this as:
\(\req(f_1,v_1,a_1)\wedge  \req(f_2,v_2,a_i)\wedge\ldots
\wedge \req(f_n,v_n,a_n)\wedge  \req(g_1,h_1,a_{n+1})
\wedge \req(g_2,h_2,a_{n+i})\wedge\ldots\wedge \req(g_n,h_n,a_{2n})
\wedge\rand(a_1,a_2,r_1)\wedge \rand(r_1,
a_3,r_2)\wedge\ldots\wedge \rand(r_{2n-2},a_{2n},i_1).\)

\smallskip
  
2.~A subformula $\SUB_2(i_2)$ such that $\SUB_2(0)$ holds true iff
two accesses  $A_1$ and $A_2$ on the $\tile$ relation, where  the output 
value of $A_1$  is equal to the value of the penultimate argument of $A_2$  
are such that their bit-vectors are in a wrong relationship. The latter just 
means that the concatenated two bit-vectors of $A_1$  are {\em not} a predecessor of 
the concatenated two bit-vectors of $A_2$.
To express this, we use the first two atoms of $Q_2$. Indeed, they are 
already linked via variable $y$. It is now just necessary to express that their
bit-vectors are wrong. To do this, we design a conjunction of atoms 
$\SUCC(\bb,\cc,\dd,\ee,s)$ for which $s=1$ iff vector $\langle \bb,\cc\rangle$   
is a numeric predecessor of $\langle \dd,\ee\rangle$, and $s=0$ otherwise. 
Then, let  $\SUB_2(i_2)= \SUCC(\bb,\cc,\dd,\ee,i_2)$.
The $\SUCC(\bb,\cc,\dd,\ee,s)$  subformula can be easily constructed by using 
purely Boolean atoms only. Briefly, we first define a
$\SUCC_i(\bb,\cc,\dd,\ee,s_i)$ formula for $1\leq i\leq 2n$  
such that $s_i=1$ iff the leading $i-1$ bits of both vectors coincide.
The $i$th bit of
$\langle \bb,\cc\rangle$ is $0$, while the $i$th bit of $\langle \dd,\ee\rangle$ is $1$, and 
all bits in positions above $i$ are $1$ in $\langle \bb,\cc\rangle$ and $0$ in $\langle \dd,\ee\rangle$.
All this is easily done with $\rand$ and $\eq$ atoms. Finally, $\SUCC$ is constructed 
by taking all $\SUCC_i$ and or-ing their $s_i$-values:
$\ror(s_1,s_2,k_1)\wedge  \ror(k_1, s_3,k_2)\wedge\ldots\wedge
\ror(k_{2n-2},s_{2n},s)$.
Note that we only need a polynomial number of atoms. 

\smallskip
  
3.~A subformula $\SUB_3(i_3)$ such that $\SUB_3(0)$ holds true iff some 
vertical or horizontal constraints are violated  or if the initial $m$ tiles are of 
wrong tile type. 
Informally, we thus need to assert in this subformula that 
there exist two tiles, such that some tiling
constraint is violated. Here we can, for example,  use the 
second and third atoms of $Q_2$.
For the horizontal constraints, we define in the obvious way subformulas that 
check that $\dd$ is the predecessor of $\fff$, that $\ee$ and  $\ggg$
are equal, and $\horiz(v,w,0)$.
The resulting truth value is or-red with a violation of the vertical constraints which is  
encoded in a similar way. To all this, we also connect disjunctively
(conjoining $\ror$ atoms) 
all possible violations of the correct tile-type of the $m$ \emph{initial tiles}.
We can use the third atom of $Q_2$ for this. 
Such a violation arises if in $\tile(w,\fff,\ggg,y',z')$ $\fff=[0]$, $\ggg=[i]<[m]$, and $w$ is 
any tile type but the correct one. This can be easily expressed using the
$\sametile$ predicate 
and Boolean operators.  

\smallskip

4.~Finally, we add to the conjunction so far $\SUB_1(i_1) \wedge \SUB_2(i_2) \wedge \SUB_3(i_3)$ a subformula $\SUB_4$ expressing that at least
 one of the bits $i_1, i_2,i_3$ must be zero: $\SUB_4 =
\rand(i_1,i_2,j)\wedge \rand(j,i_3,0)$. This concludes the construction 
 of $Q_2$.

\medskip
 
We claim that the grid is tiled iff there is an access path $p$
based on conf $\acfg$ that falsifies $Q_2$ and satisfies $Q_1$.

The only-if direction 
is quite obvious. From a correct tiling, we can easily construct a correct access path that 
starts with the two given initial tiles $\tile(t_{i_0},[0],[0],c_0,c_1)$ and $\tile(t_{i_1},[0],[1],c_1,c_2)$, and 
ends in a tile  $\tile(u,[2^n-1],[2^n-1],x,y)$, and thus satisfies $Q_1$. 
Moreover, no violation expressed by $\SUB_1(0)$ or $\SUB_2(0)$ or $\SUB_3(0)$ is present, and hence $\SUB_4$ is false over 
this access path, thus $Q_1$ is false. 

Now assume $Q_1$ is satisfied and $Q_2$ is false. 
Then, by definition of $Q_1$, $\SUB_1(0)$ and $\SUB_2(0)$ and $\SUB_3(0)$ are all false, and  
$\SUB_1(1)$ and $\SUB_2(1)$ and $\SUB_3(1)$ are true.

Given that $Q_1$ is satisfied, there exists a tile 
$\tile(u,[2^n-1],[2^n-1],x,y)$. The value $x$ must come from somewhere. Now we must consider the possibility 
that $x=c_0$, where $c_0$ is the input link value of the first tile. We would then have a fact $\tile(u,[2^n-1],[2^n-1],c_0,y)$,
which would be a ``sibling'' of the first tile $\tile(t_{i_0},[0],[0],c_0,c_1)$
which would be extremely annoying. However, this is not possible, given that this would violate the 
functional dependency from the penultimate argument to the bit-valued ones. Thus, $\SUB_1(0)$ would 
hold true, moreover, the first four atoms of $Q_2$ would be all satisfied, given that the first two 
can map to the initial facts, and the last two  to the  $\tile(u,[2^n-1],[2^n-1],c_0,y)$ and  $\tile(t_{i_0},[0],[0],c_0,c_1)$,
respectively. Hence  $Q_2$ would hold true, which is a contradiction. 

Given that here, by the semantics of access limitations, $c_0$ is the only value of its type that may not occur as an output value,  
we conclude, that the value of $x$ in  $\tile(u,[2^n-1],[2^n-1],x,y)$ must occur as the output of some tile. By the typing 
constraints and by the truth of $\SUB_2(1)$, that tile must look like  $\tile(t,[2^n-1],[2^n-2],p,x)$. By applying the 
same reasoning repeatedly, and because each access path is finite, we conclude that there exists a chain of 
$2^n$ connected accesses that connects the initial tile to a tile of the form $\tile(u,[2^n-1],[2^n-1],x,y)$ which makes $Q_1$ true. 
Note that this chain may or may not contain the second initial access
atom $\tile(t_{i_1},[0],[1],c_1,c_2)$ from $\acfg$. If it does not contain it, 
it will contain another correctly colored fact for the $\langle[0],[1]\rangle$ co-ordinates. In either case, 
due to the validity of  $\SUB_2(1)$ and $\SUB_3(1)$, this chain actually constitutes a correctly tiled grid.
\end{proof}

We now deal with upper bounds.
Chang and Li noted that for every conjunctive query (or UCQ) $Q$, and
any set of access  patterns, one can write a Monadic Datalog
query $Q_{acc}$ that represents the answers to $Q$ that
can be obtained according to the access patterns -- the intentional
predicates represent the accessible elements of each datatype, which
can be axiomatized via recursive rules corresponding to each access
method. 
One can thus show that containment of $Q$ in $Q'$ under access patterns
is reduced to containment of the  Monadic Datalog query
$Q_{acc}$ in $Q'$. Although containment between
Datalog queries is undecidable, containment of Monadic Datalog queries
is decidable (in \twoexp~\cite{cosma}) and
containment of Datalog queries in UCQs is decidable (in \threeexp~\cite{cv97}),
this does not give tight bounds for our problem.
Chaudhuri and Vardi \cite{cv97} have shown that containment of 
Monadic Datalog queries
in \emph{connected} UCQs is in
\conexptime. The queries considered there  have a head predicate with one free variable, and the connectedness
requirement is that the graph connecting atoms when they share a variable is connected -- thus the
head atom is connected to every other variable. Connectedness is a strong
condition -- it implies that in a tree-like
model one need only look for homomorphisms that lie close to the root.  

We now show a $\conexptime$  upper bound, matching our lower
bound and extending the prior results above.
From the nondeterministic polynomial-time Turing reduction
from containment to relevance
(Proposition~\ref{prop:containment-to-relevance}), we deduce
\nexptime{} membership for long-term relevance of Boolean accesses.

\begin{theorem} \label{thm:conexpmembership}
Boolean conjunctive query containment under access limitations is in
\conexptime. Long-term relevance of a Boolean access for a conjunctive query is
in \nexptime.
\end{theorem}
We now outline the proof of $\conexptime$-membership for
containment under access patterns.
Consider queries $Q$, $Q'$,
and configuration $\acfg$.

An element $n$ in an instance $\aninst$ is \emph{fresh} if it is not
in the initial configuration $\acfg$.  
Call a  homomorphism $h$ of a subquery $Q''(x)$ of $Q'$ into
an instance \emph{freshly-connected} if 
(i)~the graph whose vertices are the atoms of $Q''$ and where
there is an edge between two such atoms if they overlap in a 
variable mapped to the same fresh value by $h$ is connected;
(ii)~if a variable $y$ is mapped by $h$ to a fresh constant distinct from
$h(x)$, then $Q''$ contains all atoms of $Q'$ where $y$ appears.
Given an element $n$, a partial
homomorphism of a query $Q'(x)$ is \emph{rooted at $n$} if
it maps the distinguished variable~$x$ to~$n$ and includes only
one atom containing~$x$. 

In this proof, we assume for convenience that all values of enumerated
types mentioned in the queries are in the initial configuration -- and
thus fresh values are always of a non-enumerated type. This hypothesis
can be removed, since the part of any witness to non-containment
that involves enumerated types can always be guessed, staying within
the required bounds.

In order to get an exponential-sized
witness to non-contain\-ment, we would like to abstract an element of an instance by \emph{all} subqueries
of $Q'$ that it satisfies, or even all the freshly-connected homomorphisms.
However, this would require looking at many queries, giving
a doubly exponential bound (as in Theorem \ref{thm:2nexp}).

The following key lemma states that it suffices to look at just one
freshly-connected homomorphism.

\begin{claim} \label{claim:unique}

For each 
tree-like
instance $\aninst$ and element~$n$ in $\aninst$, and for each  
query atom $A$ that maps into an $\aninst$-atom containing $n$, 
there is a unique maximal
freshly-connected partial homomorphism rooted in~$n$ that includes atom $A$ in its domain.
\end{claim}

\begin{proof}
Consider a function $h$ mapping variables of the
atom $A$ into $\aninst$ such that $h(x)=n$. We claim
that there is only one way to extend $h$ to a 
freshly connected
homomorphism including
other atoms. Clearly, to satisfy freshness requirement~(i),
any
other atom $A'$ must include at least one other variable in common
with $A$, say $y$, mapped to $n_2$ by
$h$. To satisfy rootedness, this variable  must not be
$x$. But in a tree-like model there can be only one
other fact $F$ in $\aninst$ that contains $n_2$, and hence for
every position $p$ of $A'$ we can only map a variable
in that position to the corresponding argument of  $F$.
Furthermore, thanks to freshness requirement~(ii),
if $h$ can also be extended with $A''$ that shares variable $y'$
with $A$, it can be extended with both $A'$ and $A''$.
\end{proof}

For elements $n_1$ and $n_2$ and atom $A$, let $h^A_1$ and $h^A_2$ be the 
maximal freshly-connected partial homomorphisms given by the claim above
for $n_1$ and $n_2$, respectively.
We say two elements $n_1$ and $n_2$ in $\aninst$ are \emph{maxfresh-equivalent} if 
for every $A$ in $Q'$ there
is an isomorphism
$r$ of $\aninst$ that preserves the initial configuration,
and such that $h^A_1 \circ r= h^A_2$.

We say that fresh elements $n_1$ and $n_2$ are \emph{similar} if 
they are maxfresh-equivalent, and if $n_1$ occurs as input variable
$i$ in atom $A$, then the same is true for $n_2$  -- recall that
fresh elements in tree-like insances occur as the inputs to at most
one atom.

We now show that the maxfresh-equivalence classes
of subtrees can be determined compositionally.

\begin{lemma} \label{lem:compose} Suppose $A$
is an atom satisfied by fresh elements $n_0, n_1 \ldots n_k$ along
with configuration elements $\vec c$ in a tree-like instance
$\aninst$, with each $n_i$ in the subtree of $n_0$.
Suppose that the same is true for fresh $n'_0, n'_1 \ldots n'_k$
and $\vec c$, with $\vec c$ in the same arguments of $A$.
If each $n_i$ is maxfresh-equivalent to
$n_i'$,
then $n_0$
is maxfresh-equivalent to $n'_0$.
\end{lemma}

\begin{proof}
We show that the subtrees of $n_0$ and $n'_0$
satisfy the same subqueries of $Q'$.  Suppose
a freshly-connected subquery $Q''$ were to hold in $n_0$ with $h$ a witness. 
 Let $Q'(x)$ be
the query where $x$ is made free, and
let $h_i$ be the restriction of $h$
to the variables that are connected to $x$ and map
within subtrees of the $n_i$. This is a freshly-connected homomorphism, so
must also be realized in the subtree of $n'_i$ in $\aninst$. But then
we can extend the homomorphism to the subtree of $n'_0$ by mapping
$A$ according to the fact holding in common within $n'_0$ and $n_0$.
\end{proof}

We now show that one representative of
each similarlity class suffices 
for a counterexample model.

\begin{lemma} \label{lem:fewfresh} If $Q$ is not contained in $Q'$
under access patterns, then there is a witness instance $I$
in which no two elements are similar.
\end{lemma}

\begin{proof}
Consider an arbitrary tree-like instance $I$ satisfying the access patterns,
such that $I$ satisfies $Q$ and not $Q'$. Given distinct $n_1$ and $n_2$ that are similar.
we show that they can be identified, possibly shrinking the model. By Lemma \ref{lem:compose}
this identification preserves the similarity type, so assuming
we have proven this, we can get a single representative by induction.

Consider  the case where $n_1$ is an ancestor of $n_2$ (the case
where the two are incomparable is similar).
Consider the model $I'$ obtained by identifying $n_1$ and $n_2$, removing all
items  that lie below $n_1$ and which do not lie below $n_2$ in
the dependency graph. Let $n'$ denote the image of $n_1$ under the identification in $I'$.
$I'$ is still generated by a well-formed access path,
while $Q$ is still satisfied, since the homomorphic image of $Q$
was not modified. 
If $Q'$ were satisfied in $I'$, let $h'$ be a homomorphism witnessing this.
Clearly the image of $h'$ must include  $n'$. Let $h_{n'}$ be the maximal connected
subquery of $Q'$ that maps to a contiguous subtree rooted at $n'$.
Then up to isomorphism $h_{n'}$  is the same as the maximal fresh homomorphism
rooted at  $n_2$;
and since $n_1$ and $n_2$ are max-fresh equivalent, this
means that  $h_{n'}$ is (modulo composition with an isomorphism) the same
as $h_{n_1}$ the maximal fresh homomorphism rooted at $n_1$. 

Let $\mathit{IM}$ be  the image of $h'$. We define a mapping $f$ taking elements of $\mathit{IM}$ to $I$ as follows: nodes
in the image of $h_{n'}$ go to their isomorphic image in the image of $h_{n_1}$;
other nodes lying in both $\mathit{IM}$ and the subtree of $n'$ go to their isomorphic
image in the subtree of $n_2$, while nodes lying outside the subtree of $n'$ are
mapped to the identity. We argue that the composition of $h'$ with $f$ gives a homomorphism
of $Q'$ into $I$. Atoms all of whose elements
are in the domain of $h_{n'}$ are preserved since $n'$ and $n_1$ are max-fresh equivalent.
The properties of maximal equivalence classes guarantee
that for all other  atoms  of $Q'$ either: a) all variables that are mapped by $h'$ to  fresh
elements are mapped to elements in the subtree  of $n'$ outside of $h_{n'}$; hence such atoms
are preserved by $f$; or b) all variables that are mapped by $h'$ to fresh elements
are mapped to elements either in $h_{n'}$ or above it; such atoms are preserved, since
$f$ is an isomorphism on these elements.
Thus $Q'$ holds in $I$, a contradiction.
\end{proof}

From Lemma \ref{lem:fewfresh} we see that whenever there is  a counterexample
to containment, there is a DAG-shaped model of size the number of similarity
classes -- since this is exponential in the input, an access path that
generates it can be guessed by a \nexptime{} algorithm, and the verification
that it is a well-formed access path and witnesses $Q \wedge \neg Q'$ can
be done in polynomial time in the size of the path (albeit in $D^P$ in the size
of the queries).

\subsection{Positive Queries}
We now turn to the case where queries can use nesting
of $\vee$ and $\wedge$, but no negation. Here we will
see that the complexity of the problems becomes exponentially
harder. The upper bounds will be via a type-abstraction mechanism;
the lower bounds will again go via tiling problems, although
of a more involved sort.

\begin{theorem} \label{thm:2nexp} The problem of determining whether
query $Q$ is contained in $Q'$ under access constraints
is complete for $\cotwonexp$,  while
determining whether a Boolean access is $\lti$ for
a positive query $Q$ is complete for $\twonexp$.
\end{theorem}

\begin{proof}
Again, results about relevance are derived by the reductions of
Section~\ref{sec:containment}.

Briefly, the upper bound holds by showing that if there is
a counterexample to containment then there is one of doubly-exponential
size -- this is in turn shown by seeing that we can identify two
elements if they satisfy the same queries of size at most the
maximum of $|Q|, |Q'|$.
Hardness follows from 
reducing to the problem of tiling a corridor of width
and height doubly-exponential in $n$.
By choosing $Q'$ appropriately, we can force the model
to consist of a sequence of linked elements each of which lies at 
the root of a tree of polynomial size. Such a tree can spell out
a string of exponentially many bits, and we can further ensure
that successive occurrences trees encode a description of a doubly-exponential sized
tiling.

We first sketch the $\cotwonexp$ upper bound for containment.
We assume that the queries are connected: the full proof
can be reduced to this case, via considering connected components.

An \emph{instantiated subquery} for $Q,Q',\acfg$ is any
conjunctive query with one free variable
$x$, constants from $\acfg$, and size
at most the maximum of $|Q'|, |Q|$.

We associate with each element in an instance
the collection of instantiated subqueries  that it
satisfies, and refer to this as the type
of the  element. Note that the number
of subqueries is exponential in the
parameters, and hence the number of types
is  at most doubly exponential.

We claim that if there is an instance  $\aninst$ that
satisfies $Q \wedge \neg Q'$, there is one in which
at most exponentially elements share the same type. From this, we have
a doubly-exponential bound on the size of the model, and hence could
simply guess it, along with the accesses
that generated it, and verify that it is well-formed according
to the access patterns.

Consider an arbitrary instance $\aninst$ satisfying
$Q \wedge \neg Q'$, and let $h$ be a homomorphism
of $Q$ onto~$\aninst$. 
We can assume that the instance is tree-like as explained
in the beginning of the section. Thus
any element outside of the initial configuration and
the image of $h$ occurs
as an input to at most one access. 

Now suppose there are
two elements $n_1, n_2$, not in the image of $h$
and  having the same type, that are further than
$|Q'|$ apart in the tree.
If $n_1 \prec^* n_2$, and there are no
elements of $h(Q)$ lying strictly between them in $\prec^*$
then we
can eliminate $n_2$, replacing $n_2$ with
$n_1$ in any atom, and removing
elements lying between  $n_2$ and $n_1$ within $\prec^*$,
along with their $\prec^*$-descendants.
The $|Q'|$-type of all surviving elements are unaffected
by this replacement, and hence the resulting instance
cannot satisfy $Q'$.

If $n_1$ and $n_2$ are incomparable, and again
at distance at least $|Q'|$, then we can merge
either $n_2$ with $n_1$ or vice versa.

We can repeat this process until all elements of the tree 
with the same type are within distance $|Q'|$ of each
other, or within $|Q'|$ of the initial configuration of
$h(Q)$. Thus there will be at most exponentially many
elements with the same type, leading to an instance
of at most doubly exponential size.

We now turn to hardness, which follows from 
reducing to the problem of tiling a corridor of width
and height doubly-exponential in $n$ using
tiles having $r$ colors. 
Fix $n,r$ and vertical and horizontal constraints on
$r$ tiles.

We have a signature:
\begin{itemize}
\item $\nextcell(x, h,v,co,w)$ with $x, h, v, co $ the input variables;

\item $\nexveraddrbit(x,y,i,z)$ with $x,y,i$ input variables,
and similarly $\nexhoraddrbit(x,y,i,z)$;
\item $\iszero(c,x)$ with $c$ an input place;
\item $\isfirstzero(c,x)$ $c$ an input place;
\item $\isone(c,x)$ with $c$ an input place.
\end{itemize}

We are interested in models of a special form.
They will have a chain of elements $e_{i}: i< f$ with each  $e_{i+1}$ generated
by an access to
$\nextcell$ with input $(e_i, h, co, v)$ where:

\begin{itemize}

\item $co$ takes a value in $1 \ldots r$;

\item
corresponding to both $v$  and $h$ we  have
a tree of elements $a_\sigma: \sigma \in 2^{<n}$
where $v$ (resp., $h$) is generated by an access of the form
$\nexveraddrbit(a_1,a_0,0)$
and $a_\sigma$ for $\sigma \in 2^i$, $i<n$ is generated by an access to
$\nexveraddrbit(a_{\sigma^1}, a_{\sigma ^0},i)$.
\end{itemize}

Finally, $a_\sigma$ with $\sigma \in 2^n$ is generated
by an access to either
$\iszero(c)$, $\isfirstzero(c)$ or $\isone(c)$ with $c$ in the initial
configuration. We further require that  a call
to $\isfirstzero$ generates $a_{\sigma}$ for at most
one $\sigma \in 2^n$, and that for every $\sigma'$ representing
a binary number below $\sigma$, $a_{\sigma'}$ is generated
by a call to $\isone$.

Thus a model of this form consists of elements $e$ that are chained
together, with the predecessor of an element being the first
argument of the access that produced it. The properties
above guarantee that
$e$ is associated
with:
\begin{itemize}
\item a color, corresponding to the integer $co$ that
was used to produce $e$;
\item
a tree of $2^n$ bits  -- the frontier of elements lying
$n$ places  below the elements $v$  and $h$
used in the access producing $e$.
The value of the address corresponding to $n$-bit binary string
$\sigma$ is zero if the access generating
the element $a_{\sigma}$ was either $\isfirstzero$ or $\iszero$, and is $\isone$
otherwise. The address marked with $\isfirstzero$ will represent the
first $0$ bit -- it is useful to have this explicitly marked in order to define a successor function.
\end{itemize}

We call such a model \emph{grid-like}.
Query $Q'$ will have axioms which guarantee that a model
satisfying $\neg Q'$ is
grid-like. Most of these are straightforward, so we mention
only the axioms that deal with the $\isfirstzero$ predicate.

Given an element $e$ in the tree, we refer to a bit of $e$
as one of the elements $x$ lying $n$ accesses below the item $v$
that produced $e$.

We can write formulas:

\begin{itemize}
\item $\isbit(z,x)$ $z$ is the leaf of a bit tree for $x$;

\item $\notsamebit(z_1,z_2, x)$ stating that $z_1$ and $z_2$ are both leafs
of the bit tree of x and they are distinct leaves.
One states this by stating that the paths from $x$ to $z_1$ and
$z_2$ diverge at some point;

\item $\belowbit(z_1,z_2,x)$ stating that $z_1$ and $z_2$ are both bits of $x$ and the address
of $z_1$ comes before the address of~$z_2$;

\item $\sucbit(z_1, z_2 , x)$  stating that $z_1$ and $z_2$ are both bits of
$x$ and the address
of $z_1$ comes immediately before the address of~$z_2$. We do
this via the usual way of defining a successor function on
$n$-bit integers using unions of conjunctive queries.
\end{itemize}

Using the above macros
we will write  axioms of $Q'$ stating (i.e., forbidding in a non-contained instance):

\begin{itemize}
\item $\isfirstzero(z) \wedge \isone(z)$ (first zero must be a zero bit)

\item $\isbit(z,x) \wedge \isfirstzero(z) \wedge \isbit(z_1,x) \wedge
\isfirstzero(z_1) \wedge \notsamebit(z_1,z_2,x)$ (only one first zero)

\item $\belowbit(z_1, z_2, x) \wedge \isfirstzero(z_2) \wedge \iszero(z_1)$
(there is a zero bit below the first zero)
\end{itemize}

These will enforce the semantics of the $\isfirstzero$ predicate.

We will now want to force the grid to be of very large size.

A grid-like model is a \emph{long corridor} if as we traverse the chain of elements $e_i$ the horizontal bits cycle repeatedly through $1$ to $2^{2^n}$ and
at the end of each cycle the vertical bits increase from $0$ to a maximum of $2^{2^n}$.

We claim that we can write axioms of $Q'$ whose
negation guarantees that a model is a long corridor.

We will describe how to enforce this for the
vertical bits, with the horizontal bits done similarly. 
Formally, what we want to do is write a positive
query $\notsucc(x,y)$ that is the complement of the successor relation.
Furthermore, if we can do this, we can get $Q'$ to enforce things about
the tiling by having $Q'$ refer to $\notsucc$ at the appropriate points.

For example $Q'$ will enforce the vertical constraint on the tiles
by stating $\notsucc(x,y) \vee \bigvee_{i\leq r} \color_i(x) \wedge \color_j(y)$  
where $\color_i(x)$ is a formula
stating that $x$ is a bit of a tiling element
with  
a certain color $i \leq r$ and $i$ and $j$ range over compatible colors.

So the question is how to define $\notsucc$ with a positive query.

We define a predicate
$\sameaddressas(b,b')$  which holds of two elements in a grid-like model
iff they are both ``bits in trees'' and they are both the same bit --
i.e., they
represent the same path down a tree.
Let $\correct(x_i,y_i,x_{i-1},y_{i-1})$ abbreviate the formula:
\[
\begin{split}
&((\nexveraddrbit(x_i,u_i,z_i, x_{i-1}) \wedge  \nexveraddrbit(y_i,v_i, z_i, y_{i-1})\\
&\vee\\  
&((\nexveraddrbit(u'_i, x_i, z_i, x_{i-1}) \wedge
\nexveraddrbit(v'_i,v_i, z_i, y_{i-1})).
\end{split}\]

Now $\sameaddressas(x_1,y_1,b,b')$ is just the conjunction:
\[\begin{split}
&\correct(x_2,y_2,x_1,y_1) \wedge {}\\
&\correct(x_3,y_3,x_2,y_2)  \wedge {}\\
&\quad\ldots \\
&\correct(x_i,y_i,x_{i-1},y_{i-1}) \wedge {}\\
&\quad\ldots \\
&\correct(b,b',x_{n-1},y_{n-1}).
\end{split}
\]

Using $\sameaddressas$, we can express $\notsucc(x,y)$, e.g., by asserting
the disjunction of the three formulas:

\begin{enumerate}
\item \(\left\{
\begin{aligned}
&\isfirstzero(z,x) \wedge
\belowbit(z_1,z) \wedge \isbit(z_2,y)\\& \wedge
\sameaddressas(z_2,z_1)
\wedge \isone(z_1)
\end{aligned}\right.\)
\item \(\left\{\begin{aligned}&\isfirstzero(z,x) \wedge \isbit(z_2,y) \wedge
\sameaddressas(z_2,z) \\&\wedge \iszero(z,x) \vee
\isfirstzero(z,x)\end{aligned}\right.\)
\item \(
\left\{\begin{aligned}
&\isbit(z,x) \wedge \isbit(z',y) \wedge \sameaddressas(z,z')
\\&\wedge ((\iszero(z) \vee \isfirstzero(z)) \wedge \isone(z'))\\&\vee
((\iszero(z') \vee \isfirstzero(z')) \wedge
\isone(z))\end{aligned}\right.\) \qedhere
\end{enumerate}
\end{proof}

We note that in terms of data complexity, we still have tractability,
even in this very general case:

\begin{proposition} When the queries are fixed, the complexity 
of containment is in polynomial time. Similarly,
the complexity of $\lti$ is in polynomial time once
the query is fixed.
\end{proposition}

\begin{proof}
Again we give the argument
only for containment and use
Proposition~\ref{prop:containment-to-relevance} to conclude for relevance
(note the remark in Proposition~\ref{prop:containment-to-relevance}
about configurations being the same).
In the $\cotwonexp$ membership argument for
containment in Theorem \ref{thm:2nexp} we have shown
a witness instance in which the elements
consist of $k(Q,Q')$ elements of each type, where
the number of types is $l(Q,Q')$, hence with a constant
number of elements outside the configuration once $Q$ and
$Q'$ are fixed. The number of possible access sequences 
 is thus polynomial in the configuration, and verifying
that a sequence is well-formed and satisfies $Q \wedge \neg Q'$
can be done in polynomial time since $Q$ and $Q'$ are fixed.
\end{proof}

\section{Relations of Small Arity}
\label{sec:psp}

The argument for $\conexptime$-hardness of  containment made
heavy use of accesses with multiple inputs, in order to generate
large tree-like models.
We show that when the arity of accesses is at most binary, the complexity
does reduce.

\begin{theorem} \label{thm:pspupper}
Suppose that $Q$ is a connected positive 
Boolean query and uses only relations of arity
at most 2.  Suppose also that we have only dependent accesses.
Let $\acfg$
be a configuration and $a_0$ a Boolean access.

Then determining whether $a_0$ is $\lti$ for $Q$ in $\acfg$ can be done in $\pspace$.
As a consequence,
containment of $Q$ and $Q'$ with respect to access constraints
in this case is in $\pspace$.
\end{theorem}

\begin{proof}
In brief, we reduce the search for a witness path
to exploration of a product graph, representing
the set of types of nodes that we can reach via the access methods.
The type of a node is determined by the local neighborhood of access before
and after, which can be represented in polynomial space.
The result on containment is derived using
Proposition~\ref{prop:containment-to-relevance}.

As usual, a witness path $p$ to $\lti$ is a sequence
of accesses, beginning with the distinguished access $a_0$, such
that $Q$ is embedded into the path by a homomorphism $h$,
 but is not embedded in the
truncation of $p$. Clearly,
we can assume that $p$ is formed by taking all accesses in $h(Q)$
and then recursively throwing in, for each element $n \in h(Q)$ not
in the initial configuration $\acfg$,
the access $a^-_n$ that returned $n$, the access that returned the
input to that $a^-_n$, and so forth. For $n \in h(Q)$ not
in the initial configuration let $\chain(n)$
be the result of this process, which terminates in either an element
of the initial configuration or another element of $h(Q)$.
For two such elements of $h(Q)-\acfg$,
$n_1$ and $n_2$, say $n_1$ $\succ$ $n_2$ if $n_1$ is the first element of 
$\chain(n_2)$, and let $\preceq$ be the transitive closure of $\succ$.
Then we can re-arrange $p$ to consist of:
\[
a_0, \chain(n_1) \ldots \chain(n_k), b_1 \ldots b_l
\]
where $a_0$ is the initial access, the ordering
of the elements $n_i$ for $i \leq k$ is a linearization of
$\preceq$, and $b_1 \ldots b_l$ are additional accesses that contribute
to $h(Q)$ but do not produce new elements.

Thus a witness path can be visualized as consisting of at most $|Q|$
 linear chains, plus
at most $|Q|$ additional facts that do not introduce new elements.

Consider an automaton whose states 
are pairs $(h,f)$ where:
\begin{compactitem}
\item $h$ is a homomorphism of $Q$ into the initial configuration
$\acfg$ plus some new elements $n_i =h(x_i)$ for $x_i$ variables of $Q$;
\item $f$ is a function from $\{1,2\} \times \set{0\ldots|var(Q)|} \times var(Q)$ to
accesses, with some inputs of the accesses
marked with elements in the image of $h$.
\end{compactitem}

Informally, a state represent a description of the ``final facts''
in a witness path (those that did not introduce new elements),
plus the $|Q|$ last few and first few elements
in each of the chains.

A path $p$ and a homomorphism $h'$ of $Q$ into elements of the
path satisfies a given state $(h,f)$ if
$p$ is of the form \[a_0, \chain(n'_1) \ldots \chain(n'_j)\]
for $j \leq |Q|$, where:
\begin{compactitem}
\item $n'_i=h(x_i)$;
\item the first and the last $|Q|$ elements in each chain are
consistent with the function $f$, in that the access method
of the $j^{th}$ element after $h'(x_i)$ matches the access method
in $f(1, j, x_i)$ and has input $h'(x_i)$, iff the access $f(1,j,x_i)$ has input
$h(x_i)$;
\item  similarly for the $j^{th}$ element before
$h'(x_i)$ (with respect to $f(2,j,x_i)$).
\end{compactitem}

Given a state and
an access (with identification of input with homomorphism image),
we can determine the next state, since the new access only
impacts the final chain.

We can also determine, for
each state, if it is \emph{truncation-safe for $Q$}:
that is, whether a path $p$ and homomorphism $h'$ in this state
has the property that the path:
\[
p'= \mathrm{truncation}(p+\text{additional facts in $h'(Q)$
not witnessed in $p$})
\]
does not satisfy $Q$.

 In particular,
we claim that this is true for one $p,h'$ satisfying the state iff
it is true for all.

We explain how to check truncation safety.
If  $p'$ did satisfy $Q$, then there would be a homomorphism
into it, and since $Q$ is connected, the image
would have to lie either in
the initial configuration or within the area
of the chain within $|Q|$ places of the $n_i$; we can check
whether this is possible using the information in the state.

We do reachability in this automaton, looking for a path $s_1
\ldots s_j$ through
the automaton in which 
\begin{enumerate} 
\item each state $s_i$ traversed in the path
has the property that it is truncation safe
for $Q$;
\item the path starts with the state corresponding to the empty sequence of
accesses; 
\item the final state contained within the path has an access that
supports every $h(x_i)$.
\end{enumerate}

Since each state has polynomial size, we can explore
it in $\pspace$.
\end{proof}

We now show that the problem is \pspace-hard when the  arity is at most $3$.
We do not know if this can be improved to match the upper bound above.

\begin{proposition} \label{prop:psphard3}
Determining whether $a_0$ is $\lti$ for $Q$ in $\acfg$ is $\pspace$-hard even
when $Q$ is conjunctive and relations have arity at most $3$.
\end{proposition}

\begin{proof}
We reduce from the tiling problem for a corridor of width $n$, using
$r$ tiles, horizontal constraints $H$, vertical constraints $V$, and
initial and final tile types $i_1\dots i_n$, $f_1\dots f_n$.
We first explain how to construct queries with disjunction
for which containment is equivalent to the existence of the tiling and
discuss next how to encode this disjunction in CQs.

For all $1\leq i\leq r$ and $1\leq j\leq n$,
we define a binary predicate $C_{i,j}$
which will stands for all tiles of type $i$
at x-coordinate $j$ in the tiling: the
first attribute is the identifier of the previous tile in the
column-by-column, row-by-row, progression
and the second attribute is an identifier of the current tile. Each
of these relations has a single access methods, with input its first
attribute. All attributes share the same abstract domain.

Our first query $Q_1$ expresses that something is wrong with the tiling.
We define it as
the existential closure of a disjunction with disjuncts as follows:
\begin{description}
\item[Non-Unique Tile]~\\
$C_{i,k}(x,y) \wedge C_{j,l}(x,w)$ whenever $i \neq j$ or $k \neq
l$\\
$C_{i,k}(x,y) \wedge C_{j,l}(w,y)$ whenever $i \neq j$ or $k \neq l$;
\item[Bad Column-to-Column Progress]~\\
$C_{i,m}(x,y) \wedge C_{k,m'}(y,z)$  for $i,k \leq r, m<n, m'\neq
m+1$;
\item[Bad Row-to-Row Progress]~\\
$C_{i,n}(x,y) \wedge C_{k,m'}(y,z)$  for $i,k \leq r, m'\neq 1$;
\item[Horizontal Constraint Violations]~\\
$C_{i,m}(x,y) \wedge C_{j,m+1}(y,z)$ for $m<n$ and $i,j \not \in H$;
\item[Vertical Constraint Violations]~\\
\[\begin{split}
&C_{i,m}(x,y_1) \wedge
\bigvee_{k \leq r} C_{k,m+1}(y_1, y_2) \wedge \ldots \wedge
\\&\bigvee_{i\leq r} 
C_{i,m-1}(y_{n-1}, y_n) \wedge C_{j,m}(y_n,z) \text{ for $m \leq n, i,j
\not \in V$}\end{split}\] (with the convention that $m-1=n$ if
$m=1$).
\end{description}

$Q_2$ asserts the existence of the final row of the tiling, with
the existential closure of:
\[
C_{f_1,1}(y_0,y_1) \wedge
C_{f_2,2}(y_1, y_2) \wedge \ldots \wedge C_{f_n,n}(y_{n-1},y_n).
\]

The initial configuration $\acfg$ will consist
of facts 
\[ C_{i_1,1}(c_0,c_1) \ldots C_{i_n,n}(c_{n-1},c_n)
\]
where the $c_i$ are some distinct constants.

Suppose there is a tiling. Then we create an instance
where the domain is the cells of the tiling.
We populate predicates $C_{i,m}(x,y)$ when $x$ is the $(m-1)$th position
in some row, $y$ is the $m$th  position in the same row, for $m<n$, and
$y$ is tiled with tile type $i$; we do this for every row of the tiling other than the first.
We also populate predicates $C_{i,n}(x,y)$ when $x$ is the $n$th position
in a row, $y$ is the first  position in the next row,
and $y$ is tiled with $C_i$.
We consider an access path that navigates
the tiling starting at  the end of the first row, culminating in accesses to the final row. At the end
of this path, $Q_2$ holds, while
$Q_1$ does not hold, since any of the disjuncts holding  would
violate one of the properties of the tiling.

Conversely, suppose that $Q_1\sqsubseteq_{\acs,\acfg} Q_2$. Let
$p$ be a witness path leading to some configuration $\acfg'$.
$\acfg'$ does not satisfy $Q_1$ but satisfies $Q_2$.
It must thus have witnesses $y_0,y_1 \ldots y_n$ such that
$C_{m_1,1}(y_0,y_1) \wedge
C_{m_2,2}(y_1, y_2) \wedge \ldots\wedge C_{m_n,n}(y_{n-1},y_n)$ holds. 

First suppose that $y_0$ is in the original configuration.
If it were then since $\acfg'$ does not satisfy $Q_1$,
it can not satisfy (Non-Unique Tile), and hence
we must have $y_0=c_1$. But then by (Bad Column-to-Column Progress)
we would have that $m_1, m_2 \ldots m_n$ are also the initial
tiles, which gives a tiling.

Now suppose that $y_0$ is not in the original configuration. We
will extend the final row of tiles $(m_1 \ldots m_n)$ backward
by one tile, forming a larger partial tiling.

Since the path is well-formed, $y_0$ must have a support.
 Using the fact that $I$ does not satisfy $Q_1$, and
hence cannot satisfy (Bad Row-to-Row Progress),
this must be of the form $C_{i,n}(z_n,y_0)$ for some $z_n$ and $i$.
Using the (Vertical Constraint Violation) axiom, we see that
$(i,m_n)$ must be in the Vertical constraint. We place $i$ as the last tile
in the second-to-last row.

Reasoning similarly, we see that tracing back through the accesses
to supports gives us a traversal through a tiling: at any point
we have a partial tiling ending with the final row $r$
and beginning with an access on some element  $x$ which had been placed
in column $i$ of the tiling. If $x$ is in the initial
configuration, then we can argue that it must be $c_i$, and that adding
$c_1 \ldots c_{i-1}$ before $x$ in the row completes a tiling.
If $x$ is not in the initial configuration, then its support must
be consistent with both the horizontal constraint and the vertical constraint,
and hence we can continue.

We cannot continue indefinitely without reaching the initial configuration,
since at each step that we continue we add a fresh element from the path to the tiling.
Hence eventually we must complete the tiling with the initial row.

We now comment on how union can be eliminated, by adding to the arity.
The technique is the same as the one used in the proof of
\conexptime-hardness of containment in general
(Theorem~\ref{thm:conj-nexptime-hard}): we encode the disjunctive
constraints of $Q_1$ as a conjunction of $\mathit{And}$, $\mathit{Or}$,
$\mathit{Eq}$ atoms with Boolean constraints, together with atoms witnessing
two (non-necessary distinct) lines of the tiling.
The new predicates have arity 3.
\end{proof}

\section{Related Work}
\label{sec:related}

We 
overview  the existing literature on answering
queries in the presence of limited access patterns, highlighting
the differentiators of our approach using the vocabulary of
Section~\ref{sec:prelim}. The
problem of querying under access restrictions
was  originally  motivated by  access
to built-in-predicates with infinitely many tuples (such as~$<$) --  which is only
reasonable if all variables are bound -- and by the  desire
 to access
relations only on their indexed attributes (see
\cite{ullman1989principles}, chapter~12). More recently, the rise of
data-integration systems with sources whose content are accessible
through interfaces with limited
capabilities~\cite{rajaraman1995answering} has been the main
driver of  interest in the subject; the most well-known
example is the querying of deep Web
sources~\cite{he2007accessing} accessible through
Web forms. 
Research efforts on deep Web exploration~\cite{he2007accessing} and
on the use of Web services to complement extensional knowledge
bases~\cite{PredaKSNYW10} are practical settings where the
notion of dynamic relevance of a Web source is fundamental.

With a few exceptions that will be noted, most
of the existing work focuses on static analysis 
for dependent accesses. By static we mean that they seek a means to answer the query
\emph{ab initio} that does not consider the configuration or adapt to it.
By dependent we mean, as in the second part of this paper, that it is
impossible to guess a constant to be used in a bound access.
Typically accesses are
also supposed to be exact and not merely sound. This last
limitation is unrealistic in the case of deep Web sources, where a given
source
will often have only partial knowledge over some
data collection.
In contrast, our results
give new bounds on static problems, but also
consider dynamic relevance. We allow a collection
of sound accesses that can be dependent or independent.
Queries considered in the literature
are usually conjunctive, but work on query answerability
and rewriting has also considered richer query
languages, such as unions of conjunctive queries
(UCQs)~\cite{li2003computing}, UCQs with
negations~\cite{nash2004edbt,deutsch2007rewriting}, or even first-order
logic~\cite{nash2004pods}. In a few
cases~\cite{rajaraman1995answering,duschka1997recursive,deutsch2007rewriting},
existing work assumes that both queries and sources are expressed as
\emph{views} over a global schema, and the limited access problem is
combined to that of answering queries using
views~\cite{answering2001halevy}. In the other cases, the query is
supposed to be directly expressed in terms of the source relations, as we
have done in this work.

\paragraph*{Static analysis} The first study of query answering when
sources have limited patterns is by Rajamaran, Sagiv, and
Ullman~\cite{rajaraman1995answering}. Given a conjunctive query over a
global schema and a set of views over the same schema, with exact
dependent access patterns, they show that determining whether there
exists a conjunctive query plan over the views that is equivalent to the
original query and respects the access patterns is
\complexity{NP}-complete. This is based on the observation that the size
of a query plan can be bounded by the size of the query: one can just
keep in the plan subgoals that either are mapped to one of the subgoal of
the query, or provide an initial binding for one of the variables of the
queries.

Duschka, Genesereth, and Levy study in~\cite{duschka1997recursive} the
general problem of answering queries using views. They solve this by
constructing a Datalog query plan formed of \emph{inverse rules} obtained
from the source descriptions, a plan computing the
maximally contained answer to a query. Although
this approach was geared towards data integration without limitations,
the same work 
extends it to incorporate limited access patterns on sources in
a straightforward manner.

Li and Chang~\cite{li2001answering} propose a static query planning
framework based on Datalog for getting the maximally contained answer to
a query with exact dependent access pattern; the query language
considered is a proper subset of UCQs, with only
natural joins allowed. In the follow-up 
work~\cite{li2003computing}, the query language is lifted to UCQs,
and Li shows that testing the existence of an exact
query rewriting is \complexity{NP}-complete, which can be seen as an
extension of the result from~\cite{rajaraman1995answering}, from CQs to
UCQs (though views are not considered in~\cite{li2003computing}). That
article also proposes a dynamic approach when an exact query rewriting
does not exist, which we discuss further.

Nash and Lud\"ascher \cite{nash2004edbt} extend the results
of~\cite{li2003computing} to the case of UCQs
with negations. The complexity of the exact rewriting problem is
reduced to standard query containment and thus becomes in
this case \pitp-complete. Other query languages are also considered
in~\cite{nash2004pods}, up to first-order-logic, for which the rewriting problem
becomes undecidable. Deutsch, Ludäscher, and Nash add
in~\cite{deutsch2007rewriting} views and constraints (in the form of
weakly acyclic tuple-generating dependencies) to the setting of dependent
exact patterns. Using the chase procedure, they show that the exact
rewriting problem remains \pitp-complete in the presence of a large
class of integrity constraints,  and they provide algorithms for
obtaining both the maximally contained answer and the minimal
containing static plan.

Finally, still in the case of dependent exact accesses and for
conjunctive queries, Calì and Martinenghi \cite{cali2008querying} build
on the query planning framework of~\cite{li2001answering} and show how to
obtain a query plan for maximally contained answer that is minimal in
terms of the number of accesses made to the sources.

\paragraph*{Dynamic computation of maximal answers}
Some works referenced in the previous
paragraphs also consider dynamic, runtime, aspects of the problem, i.e.,
taking into account the current configuration. Thus,
\cite{li2003computing} provides an algorithm that finds the complete
answer to a query under dependent exact accesses whenever possible, even
if an exact query rewriting plan cannot be obtained. This is based on a
recursive, exhaustive, enumeration of all constants that can be retrieved
from sources, using the techniques of the inverse rule
algorithm~\cite{duschka1997recursive}. The algorithm has no optimality
guarantee, since no check is made for the relevance of an access to the
query, for any notion of relevance. An extension to UCQs with negation is
proposed in~\cite{nash2004edbt}, with a very similar approach.

\paragraph*{Dynamic relevance}
To our knowledge, the only work to consider the dynamic relevance of a
set of accesses with binding patterns is by Calì, Calvanese, and
Martinenghi~\cite{cali2009dynamic}. They define an access with a binding
as \emph{dynamically relevant} under a set of constraints (functional
dependencies and a very restricted version of inclusion dependencies) for
a given configuration if this access can produce new tuples. They show that
dynamic relevance can be decided in polynomial time. Note that the fact
that a source has limited access patterns does not play any role
here since one
only considers a given binding and disregards all other accesses.
Furthermore, there is no query involved.

\paragraph*{Related work on query containment} We want to conclude this section by
mentioning a few other works that are not dealing with answering queries
under binding patterns \emph{per se} but are still pertinent to the problem
studied in this paper. We have already compared in detail in
Section~\ref{sec:containment} our work with the
complexity analysis~\cite{cali2008conjunctive}
of query containment under access limitations -- in brief, our
results generalize the upper bounds to a richer model, provide
matching lower bounds, and give bounds for larger collections of queries.
However the
arguments used in proofs of our upper bound results of this paper rely
heavily on the crayfish chase procedure
described in this work.

Chaudhuri and Vardi \cite{cv94, cv97} consider the problem of containment of
Datalog queries in unions of conjunctive queries -- it is
easily seen (e.g. from the Datalog-based approaches to limited
access patterns mentioned earlier) that
this problem subsumes
containment under access patterns. Indeed, containment
under access patterns is subsumed by  containment of Monadic
Datalog in UCQs, a problem shown by \cite{cv94} to be in $\conexptime$
in special cases (e.g. connected queries without constants): the 
upper-bounds rely on the ability to make models tree-like, as 
ours do.
In the case of binary accesses, containment under access limitations
can be reduced to containment of a \emph{path query} in a conjunctive
query. This problem is studied in \cite{florescupath, calvanesepath}
which together give an $\expspace$-bound for containment of path queries --
as our results show, these  results give neither tight bounds for the 
binary case or the best lower bound for the general case.

\paragraph*{Other related work}
Finally, Abiteboul, Bourhis, and
Marinoiu~\cite{abiteboul2009satisfiability} consider dynamic relevance of
a service call to a query in the framework of ActiveXML (XML documents
with service calls). A service call is dynamically relevant if it can
produce new parts of the tree that will eventually change the query
result. They show in particular that non-relevance is in \sigmatp{} (an
unpublished extension of their work shows \sigmatp-hardness also by
reduction from the critical tuple problem~\cite{miklausuciu}).
Though the framework is different, their notion of relevance is
close in spirit to our notion of long-term relevance.

\section{Conclusion}

We investigated here the problems of 
analyzing the access paths that can originate
from a particular data access. When accesses are not tightly coupled,
the problem is closely-related to  reasoning whether a tuple is ``critical''
to a given query result. Here we have fairly tight complexity bounds, although
admittedly no algorithms that are promising from a practical perspective as yet.

In the setting where accesses are dependent on one another,  we have
shown a tight connection between relevance problems and containment under
access limitations for Boolean accesses. We have shown new bounds for both the relevance and
containment problems for conjunctive and positive queries. However, there
are still many open issues regarding complexity. For low arity
we do not have tight bounds on containment or
relevance. For arbitrary arity we have tight bounds for CQs and for
positive queries, but we do not know if our lower bounds for positive
queries also hold for positive queries of restricted forms (e.g., UCQs).
We believe that all of our results for containment can be extended 
to relevance of non-Boolean accesses, using the same proofs, but we leave 
this for future work.

Of course, this work is a small step in understanding the possible
paths from a database configuration that obey a given set of
semantic restrictions. We believe the techniques applied here
can be used  to determine the complexity of dynamic
relevance and containment under access restrictions in
the presence of integrity constraints and views.
We are also studying the impact of paths on not just
the certain answers, but also on consistent answers.
This would be especially important if access methods were assumed to be
exact and not merely sound as here.

\section{Acknowledgments}
M.~Benedikt is
supported in part by EP/G004021/1
and EP/H017690/1 of the Engineering and Physical Sciences Research Council
UK
and in part by EC FP7-ICT-233599.
G.~Gottlob's and P.~Senellart's research leading to these results has
received funding
from the European Research Council under the European Community's
Seventh Framework Programme (FP7/2007-2013) / ERC grant agreements
246858 \textit{DIADEM} and 226513 \textit{Webdam}, respectively.
We thank Pierre Bourhis for great assistance with an earlier draft of the
paper.

\bibliographystyle{abbrv}
\bibliography{limited}

\end{document}